\documentclass[orivec,conference]{llncs}
\usepackage[cmex10]{amsmath}
\usepackage{stmaryrd}
\usepackage{amsfonts}
\usepackage{amssymb}
\usepackage{bbold}
\usepackage{wrapfig}
\usepackage{graphicx}
\usepackage{enumerate}
\usepackage{multirow}
\usepackage{xspace}
\usepackage{cancel}
\usepackage{enumitem}
\usepackage[hidelinks]{hyperref}
\usepackage{caption}
\usepackage{xcolor,colortbl}
\usepackage{esvect}
\newcommand{\ict}{\ti{IC3}\xspace}
\newcommand{\data}{\mbox{$\mi{Data}$}\xspace}
\newcommand{\fifo}{\mbox{$\mi{Fifo}$}\xspace}

\newcommand{\prd}{\ti{PrvRed}\xspace}

\newcommand{\lrn}{\ti{Lrn}\xspace}

\newcommand{\apqe}{\mbox{$\mi{DS}$-$\mi{PQE}$}\xspace}

\newcommand{\Apqe}{\mbox{$\mi{START}$}\xspace}
\newcommand{\cad}{\mbox{$\mi{CADET}$}\xspace}

\newcommand{\Mapqe}{\mbox{$\mi{START}^*$}\xspace}

\newcommand{\Npqe}{\ti{START}\xspace}

\newcommand{\bm}[1]{{\mbox{\boldmath $#1$}}}
\newcommand{\Bm}[1]{{\boldmath $#1$}}

\newcommand{\spe}{\mbox{$\mi{Sp}$}\xspace}

\newcommand{\imp}{\Rightarrow} 
 \newcommand{\ie}{\emph{i.e.},
} 
\newcommand{\pnt}[1]{\mbox{$\vec{#1}$}\xspace}
\newcommand{\ppnt}[2]{\mbox{$\vv{#1}_{#2}$}}
\newcommand{\pqnt}[2]{\mbox{$\vv{#1}_{#2}$}}
\newcommand{\pent}[2]{\mbox{$\vv{#1}\!_{#2}$}}

\newcommand{\cof}[2]{\mbox{$#1_{\vec{#2}}$}}
\newcommand{\ccof}[3]{\mbox{$#1_{\vv{#2}\!_{#3}}$}}

\newcommand{\V}[1]{\mbox{$\mathit{Vars}(#1)$}}
\newcommand{\Va}[1]{\mbox{$\mi{Vars}(\vec{#1})$}}

\newcommand{\s}[1]{\mbox{$\{#1\}$}}

\newcommand{\nGz}[2]{$G_{non-\{z\}}$}

\newcommand{\prr}[1]{\mi{Prev}(\boldsymbol{q})}

\newcommand{\mi}[1]{\mathit{#1}}
\newcommand{\ti}[1]{\textit{#1}}
\newcommand{\tb}[1]{\textbf{#1}}

\newcommand{\ttt}{\>\>\>}

\newcommand{\Tt}{\>\>}
\newcommand{\Sub}[2]{\mbox{$\mi{#1}_{\mi{#2}}$}}

\newcommand{\Sup}[2]{\mbox{$#1^\mi{#2}$}}

\newcommand{\prob}[2]{\mbox{$\exists{#1} [#2]$}}

\newcommand{\Comment}[1]{}

\newcommand{\pqe}[4]{$\prob{#1}{#3 \wedge #4} \equiv #2 \wedge \prob{#1}{#4}$}

\newcommand{\Abs}[1]{\mbox{$\mathbb{S}_{#1}$}}

\begin{document}

% STEM (Single Target clausE replaceMent)
% STEAM (Single Target clausE replAceMent)
% STREAM (Single TaRget clausE replAceMent)
% START (proving Single Target clAuse RedundanT)
% START (Single TARgeT)
% STUD (Single Target claUse reDundancy)
% STUN (Single Target claUse reduNdancy)
% STAR (Single TArget Redundancy)
% STUDY (Single Target claUse reDundancY)
% PRUNE (PaRtial qUaNtifier Elimination)
% PUTIN (Partial qUanTifier elIminatioN) :-):-)
% PRICE (PRovIng Clause rEdundancy)
%\title{On Partial Quantifier Elimination}
\title{Partial Quantifier Elimination By Certificate Clauses}
%\title{Partial Quantifier Elimination By Redundancy Based Reasoning}
%\title{Partial Quantifier Elimination By Computing Clause Redundancy}
%\title{Partial Quantifier Elimination By Computing Redundancy}
%\title{Partial Quantifier Elimination By Redundancy Propagation}
%\title{Partial Quantifier Elimination By Clause Replacement}

%% \begin{center}
%%   \Large\bfseries\boldmath
%%   Generation Of Complete Test Sets
%% \end{center}

%% \author{\IEEEauthorblockN{Eugene Goldberg} 
%% \IEEEauthorblockA{
%% email:
%% eu.goldberg@gmail.com}}

\author{Eugene Goldberg}
\institute{\email{eu.goldberg@gmail.com}}
%% \author{}
%% \institute{}

\maketitle

%\vspace{-10pt}
\begin{abstract}
In this report\footnote{Some new results have been produced after publishing this report. We
discuss them in Section~\ref{sec:retrosp}.

}, we study partial
quantifier elimination (PQE) for propositional CNF formulas. PQE is a
generalization of quantifier elimination where one can limit the set
of clauses taken out of the scope of quantifiers to a small subset of
target clauses. The appeal of PQE is twofold. First, PQE can be
dramatically simpler than full quantifier elimination. Second, many
verification problems (e.g. equivalence checking and model checking)
can be solved in terms of PQE.  Our approach is based on deriving
clauses depending only on unquantified variables that make the target
clauses \ti{redundant}.  Proving redundancy of a target clause is done
by construction of a ``certificate'' clause implying the former. We
describe a PQE algorithm called \Npqe that employs the approach
above. To evaluate \Npqe, we apply it to invariant generation for a
sequential circuit $N$. The goal of invariant generation is to find an
\ti{unwanted} invariant of $N$ proving unreachability of a state that
is supposed to be reachable.  If $N$ has an unwanted invariant, it is
buggy.  Our experiments with FIFO buffers and HWMCC-13 benchmarks
suggest that \Npqe can be used for detecting bugs that are hard to
find by existing methods.
\end{abstract}

\section{Introduction}
\label{sec:intro}
In this paper, we consider the following problem. Let $F_1(X,Y),
F_2(X,Y)$ be propositional formulas in conjunctive normal form
(CNF)\footnote{Every formula is a propositional CNF formula unless otherwise
stated. Given a CNF formula $F$ represented as the conjunction of
clauses $C_1 \wedge \dots \wedge C_k$, we will also consider $F$ as
the \ti{set} of clauses \s{C_1,\dots,C_k}.
} where $X,Y$ are sets of
variables. Given \prob{X}{F_1\!\wedge\!F_2}, find a quantifier-free
formula $F^*_1(Y)$ such that $\prob{X}{F_1 \wedge F_2}\equiv
F^*_1\wedge\prob{X}{F_2}$.  In contrast to quantifier elimination
(QE), only a part of the formula gets ``unquantified'' here.  So, this
problem is called \ti{partial} QE (PQE)~\cite{hvc-14,pqe_page}. We
will refer to $F^*_1$ as a \ti{solution} to PQE.  Like SAT, PQE is a
way to cope with the complexity of QE. But in contrast to SAT that is
a \ti{special} case of QE (where all variables are quantified), PQE
\ti{generalizes} QE. The latter is just a special case of PQE where
$F_2\!=\!\emptyset$ and the entire formula is
unquantified.

The appeal of PQE is twofold. First, it can be much more efficient
than QE if $F_1$ is a \ti{small} part of the formula. Second, PQE
facilitates the development of new approaches to various verification
problems like SAT~\cite{south_korea}, equivalence
checking~\cite{fmcad16}, model checking~\cite{tech_rep_pc_lor} and so
on.

We solve PQE by \ti{redundancy based reasoning}.  Its introduction is
motivated by the following observations. First, $F_1 \wedge F_2 \imp
F^*_1$ and $F^*_1 \wedge \prob{X}{F_1 \wedge F_2}$ $\equiv$ $F^*_1
\wedge \prob{X}{ F_2}$. Thus, a formula $F^*_1$ implied by $F_1 \wedge
F_2$ becomes a solution as soon as $F^*_1$ makes the clauses of $F_1$
\ti{redundant}.  Second, one can prove clauses of $F_1$
redundant\footnote{By ''proving a clause $C$ redundant'', we mean showing that $C$ is
redundant after adding (if necessary) some new clauses.
} one by one. The redundancy of a
clause $C \in F_1$ can be proved by using $(F_1 \cup F_2) \setminus
\s{C}$ to derive a clause $K$ implying $C$. We refer to $K$ as a
\ti{certificate clause}. Importantly, one can produce $K$ even if
$(F_1 \cup F_2) \setminus \s{C}$ does not imply $C$. This becomes
possible if one allows generation of clauses preserving
\ti{equisatisfiability} rather than equivalence.

We implement redundancy based reasoning in a PQE algorithm called
\linebreak\Apqe, an abbreviation of Single TARgeT. At any given
moment, \Apqe proves redundancy of only one clause (hence the name
``single target'').  \Apqe builds the certificate $K$ above by
resolving ``local'' certificate clauses implying the clause $C$ in
subspaces. Proving redundancy of $C$ in subspaces where $F_1 \wedge
F_2$ is unsatisfiable, in general, requires adding new clauses to $F_1
\wedge F_2$. The added clauses depending only on unquantified
variables form a solution $F^*_1$ to the PQE problem. \Apqe is
somewhat similar to a SAT-solver with conflict driven learning. A
major difference here is that \Apqe backtracks as soon as the target
clause is proved redundant in the current subspace (even if no
conflict occurred).

In this paper, we apply \Apqe to the problem of invariant generation
that has no general solution yet.  We use invariant generation for bug
detection as described below. Our objective here is to provide a
``proof of concept'' for PQE i.e. to give some experimental evidence
that PQE is an important direction for research. Let $N$ be a
sequential circuit to verify. As far as reachable states of $N$ are
concerned, one can have bugs of two kinds. A bug of the first kind
occurs if a bad state is reachable in $N$. A bug of the second kind
takes place if a required good state (i.e. one that is supposed to be
reachable) is \ti{unreachable} in $N$. One excludes bugs of the first
kind by checking that a set of desired invariants holds. The challenge
here is that these invariants may be hard to prove. Bugs of the second
kind are currently identified either by testing or by checking if $N$
has an \ti{unwanted} invariant. An invariant $P$ of $N$ is unwanted if
a required good state falsifies $P$ and so is unreachable in $N$.  If
$P$ holds for $N$, the latter has a bug of the second kind. The
unwanted invariants to check are currently generated manually i.e. are
guessed. So, one can easily overlook a bug of the second kind. The
main challenge here is to \ti{find} an unwanted invariant that holds
rather than the hardness of proving it true.

The main body of this paper is structured as follows. (Some additional
information is given in the appendix.)  Sections~\ref{sec:basic} and
\ref{sec:es_impl} provide some basic definitions. A description of
\Apqe is given in Sections~\ref{sec:exmp1}-\ref{sec:trg_unit}. In
Section~\ref{sec:inv_gen}, we show that PQE can be used to
\ti{automatically} generate invariants to check for being unwanted. In
Section~\ref{sec:exper1}, we use \Apqe to detect a bug of the second
kind in a FIFO buffer that is hard to find by existing
methods\footnote{In Appendix~\ref{app:exper2}, we present results
  showing that \Apqe is efficient enough to generate invariants of
  HWMCC-13 benchmarks. We also give evidence that PQE can be
  dramatically more efficient than QE.}.  Section~\ref{sec:be_test}
explains how to decide if an invariant is unwanted via test
generation. Sections~\ref{sec:bg} and \ref{sec:concl} provide some
background and conclusions. Finally, in Section~\ref{sec:retrosp}, we
discuss some new results obtained after the publication of this
report.

%\section{Redundant Clauses, Boundary Points, And Quantifier Elimination}
\vspace{-5pt}
\section{Basic Definitions}
%\flushbottom
%\vspace{-10pt}
\label{sec:basic}
We assume that every formula is in CNF unless otherwise stated. In
this section, when we say ``formula'' without mentioning quantifiers,
we mean ``a quantifier-free formula''.

%
%  Vars(F)
%
%\vspace{-10pt}
\begin{definition}
  \label{def:vars} Let $F$ be a formula. Then \bm{\V{F}} denotes the
set of variables of $F$ and \bm{\V{\prob{X}{F}}} denotes
$\V{F}\!\setminus\!X$.
\end{definition}

%
% Assignment
%
%\vspace{-10pt}
\begin{definition}
Let $V$ be a set of variables. An \tb{assignment} \pnt{q} to $V$ is a
mapping $V'~\rightarrow \s{0,1}$ where $V' \subseteq V$.  We will
denote the set of variables assigned in \pnt{q}~~as \bm{\Va{q}}. We will
refer to \pnt{q} as a \tb{full assignment} to $V$ if $\Va{q}=V$. We
will denote as \bm{\pnt{q} \subseteq \pnt{r}} the fact that a) $\Va{q}
\subseteq \Va{r}$ and b) every variable of \Va{q} has the same value
in \pnt{q} and \pnt{r}.
\end{definition}
%
%  Cofactor
%
%\vspace{-10pt}
\begin{definition}
\label{def:cofactor}
Let $C$ be a \tb{clause} (i.e. a disjunction of literals). Let $H$ be
a formula that may have quantifiers, and \pnt{q} be an assignment to
\V{H}.  If $C$ is satisfied by \pnt{q}, then \bm{\cof{C}{q} \equiv
  1}. Otherwise, \bm{\cof{C}{q}} is the clause obtained from $C$ by
removing all literals falsified by \pnt{q}. Denote by \bm{\cof{H}{q}}
the formula obtained from $H$ by removing the clauses satisfied by
\pnt{q} and replacing every clause $C$ unsatisfied by \pnt{q} with
\cof{C}{q}.
\end{definition}

%
% X-clause, non-X-clause
%
%\vspace{-10pt}
\begin{definition}
  \label{def:Xcls}
Given a formula \prob{X}{F(X,Y)}, a clause $C$ of $F$ is called a
\tb{quantified clause} if \V{C} $\cap~X~\neq~\emptyset$. If $\V{C}
\cap X = \emptyset$, the clause $C$ depends only on free
i.e. unquantified variables of\, $F$ and is called a \tb{free clause}.
\end{definition}

%
%   Formula equivalence
%
%\vspace{-10pt}
\begin{definition}
\label{def:formula-equiv}
Let $G, H$ be formulas that may have existential quantifiers. We say
that $G, H$ are \tb{equivalent}, written \bm{G \equiv H}, if $\cof{G}{q} =
\cof{H}{q}$ for all full assignments \pnt{q} to $\V{G} \cup \V{H}$.
% Notice that if $\Va{q} \subseteq \V{\exists X
%  [F]}$, then $\cof{(\exists X [F])}{q} \equiv \exists X [\cof{F}{q}]$.
\end{definition}

%
% redundant clauses
%
%\vspace{-10pt}
\begin{definition}
\label{def:red_cls}
Let $F$ be a formula and $G \subseteq F$ and $G \neq
\emptyset$. Formula $G$ is \textbf{redundant in} \bm{\prob{X}{F}} if
$\prob{X}{F} \equiv \prob{X}{F \setminus G}$.
\end{definition}

%
%   PQE problem
%
%\vspace{-10pt}
\begin{definition}
 \label{def:pqe_prob} Given a formula \prob{X}{F_1(X,Y)\!\wedge\!
 F_2(X,Y)}, the \tb{Partial Quantifier Elimination} (\tb{PQE}) problem
 is to find $F^*_1(Y)$ such that \Bm{\prob{X}{F_1\wedge F_2}\equiv
 F^*_1\wedge\prob{X}{F_2}}.  (So, PQE takes $F_1$ out of the scope of
 quantifiers.)  $F^*_1$ is called a \tb{solution} to PQE. The case of
 PQE where $F_2=\emptyset$ is called \tb{Quantifier Elimination}
 (\tb{QE}).
\end{definition}
\begin{remark}
\label{rem:noise}
Let $C$ be a clause of a solution $F^*_1$ to the PQE problem above.
If $F_2$ implies $C$, then $F^*_1 \setminus \s{C}$ is a solution too.
\end{remark}

%\input{i2nterp_and_pqe}
%\input{a2ppl_of_pqe}
%\section{Introducing Implication Via Redundancy}
%\section{Extending The Notion Of Implication}
\section{Extended Implication And Blocked Clauses}
\label{sec:es_impl}
One can introduce the notion of implication via that of redundancy.
Namely, $F \imp G$, iff $G$ is redundant in $F \wedge G$ i.e. iff $F
\wedge G \equiv F$.  We use this idea to extend the notion of
implication via redundancy in a \ti{quantified} formula.
%
% Definition: es-implication with respect to Y
%
\begin{definition}
\label{def:es_impl}
Let $F(X,Y)$ and $G(X,Y)$ be formulas and $G$ be redundant in
\prob{X}{F \wedge G} i.e. \prob{X}{F \wedge G} $\equiv$
\prob{X}{F}. Then \cof{(F \wedge G)}{y} and \cof{F}{y} are
equisatisfiable for every full assignment \pnt{y} to $Y$.  So, we will
say that $F$ \tb{es-implies} $G$ in \prob{X}{F \wedge G}. (Here ``es''
stands for ``equisatisfiability''.)  A clause $C$ is called an
\tb{es-clause} in \prob{X}{F \wedge C} if $F$ es-implies $C$ in
\prob{X}{F \wedge C}.  One can view es-implication as a weaker version
of regular implication.
\end{definition}

Note that if $F$ implies $G$, then $F$ also es-implies $G$ in
\prob{X}{F \wedge G}. However, the converse is not true. We will say
that $F$ es-implies $G$ without mentioning the formula \prob{X}{F
  \wedge G} if the latter is clear from the context.
%
% Resolution of clauses
%
\begin{definition}
\label{def:resol}
Let clauses $C'$,$C''$ have opposite literals of exactly one variable
$w\!\in\!\V{C'}\!\cap\!\V{C''}$.  Then $C'$,$C''$ are called
\tb{resolvable} on~$w$.  The clause $C$ having all literals of
$C',C''$ but those of $w$ is called the \tb{resolvent} of
$C'$,$C''$. The clause $C$ is said to be obtained by \tb{resolution}
on $w$.
\end{definition}

Clauses $C',C''$ having opposite literals of more than one variable
are considered \tb{unresolvable} to avoid producing a tautologous
resolvent $C$ (i.e. $C \equiv 1$).

% Definition: blocked clause
%
\begin{definition}
  \label{def:blk_cls}
Given a formula \prob{X}{F(X,Y)}, let $C$ be a clause of $F$. Let $G$
be the set of clauses of $F$ resolvable with $C$ on a variable $w \in
X$. Let $w=b$ satisfy $C$, where $b \in \s{0,1}$. We will call $C$
\tb{blocked} in \prob{X}{F} at $w$ if $G$ is redundant in \prob{X}{F}
in subspace $w=b$ (i.e. if $G_{w=b}$ is redundant in \prob{X}{F_{w=b}}).
\end{definition}
%
% remark
%
\begin{remark}
\label{rem:blk_cls}
Note that if $G = \emptyset$ or the clauses of $G$ are removed from
\prob{X}{F} as redundant, $C$ meets the \ti{original} definition of a
blocked clause~\cite{blocked_clause}. Definition~\ref{def:blk_cls}
allows to declare $C$ blocked \ti{without} removing clauses of $G$ if
a proof of their redundancy in \prob{X}{F} is available. This feature
is used by our PQE-solver \Apqe (see Remark~\ref{rem:loc_rem} of
Section~\ref{sec:trg_unit}).
\end{remark}

%
% Proposition on the blocked clause
%

\begin{proposition}
\label{prop:blk_clause}
Given a formula \prob{X}{F(X,Y)}, let $C$ be a clause blocked in
\prob{X}{F} at $w \in X$. Then $C$ is redundant in \prob{X}{F}
i.e. \prob{X}{F}\,$\equiv$ \prob{X}{F \setminus \s{C}}. So,
$C$ is es-implied by $F \setminus \s{C}$ in \prob{X}{F}.
\end{proposition}

Proofs of the propositions are given in Appendix~\ref{app:proofs}.

%\input{r0ed_based_reasoning}
%\vspace{-15pt}
%\section{A Simple Example Of Solving PQE}
\section{A Simple Example Of How \Apqe Operates}
\label{sec:exmp1}

In this paper, we introduce a PQE algorithm called \Apqe (an
abbreviation of Single TARgeT).  In this section, we give a taste of
\Apqe by a simple example. Figure~\ref{fig:exmp1} describes how \Apqe
operates on the problem shown in lines 1-6. (Figure~\ref{fig:exmp1}
and Figures \ref{fig:exmp2},\ref{fig:exmp3},\ref{fig:exmp5} of the
appendix are built using a version of \Apqe generating execution
traces.  A Linux binary of this version can be downloaded
from~\cite{start_binary}.)

%
% first example of start's operation
%
\setlength{\intextsep}{4pt}
\setlength{\textfloatsep}{10pt}
\begin{wrapfigure}{l}{1.4in}
%\begin{figure}
%\begin{center}
\small
%\scriptsize
%\vspace{5pt}
\begin{tabbing}
aaaa\=bb\=cc\=dd\= \kill

\scriptsize{1}\>  Find $F^*_1(Y)$ such that \\
\scriptsize{2}\>  \prob{X}{F_1\!\wedge\!F_2} $\equiv\!F^*_1\wedge$\prob{X}{F_2}  \\
\scriptsize{3}\> $Y=\s{y_1}$, $X=\s{x_2,x_3}$\\
\scriptsize{4}\> $F_1 = \s{C_1}$, $C_1=\overline{x}_2 \vee x_3$\\
\scriptsize{5}\> $F_2\!=\!\s{C_2,C_3}$, $C_2\!=\!y_1\!\vee\!x_2$,\\
\scriptsize{6}\Tt $C_3=y_1 \vee \overline{x}_3$ \\[4pt]

\scriptsize{7}\>  pick. $C_1 \in F_1$ to prove red. \\[4pt]
\scriptsize{8}\>    --- call \prd ---\\
\scriptsize{9}\> decision: $y_1=0$  at level 1 \\
\scriptsize{10}\> \ti{BCP}:$(C_2\!:x_2\!=\!1)(C_3\!:x_3\!=\!0)$ \\[4pt]
\scriptsize{11}\> LEAF: conflict at  level 1\\
\scriptsize{12}\> $C_1\!=\!\overline{x}_2\!\vee\!x_3$ is falsified\\
\scriptsize{13}\> gen. particip. cert. $C_4 = y_1$ \\
\scriptsize{14}\Tt $R_1=\mi{Res}(C_1,C_2,x_2)$,\\
\scriptsize{15}\Tt $C_4=\mi{Res}(R_1,C_3,x_3)$ \\
\scriptsize{16}\> $F_1 = F_1 \cup \s{C_4}$\\ [4pt]

\scriptsize{17}\> backtracking to level 0 \\
\scriptsize{18}\> \ti{BCP}:  $(C_4:\!y_1=1)$ \\[4pt]

\scriptsize{19}\> LEAF: $C_1$ is blocked at $x_2$\\
\scriptsize{20}\> (since $C_2$ is sat. by $y_1$ = 1)\\
\scriptsize{21}\> $K_1\!=\!\overline{y}_1\!\vee\overline{x}_2$ is the init. cert.\\
\scriptsize{22}\> $K_2\!=\!\overline{x}_2$ is the final cert.  \\
\scriptsize{23}\Tt $K_2\!=\!\mi{Res}(K_1,\!C_4,\!y_1)$ \\
\scriptsize{24}\> $K_1,K_2$ are witness certs. \\
\scriptsize{25}\>  not added to $F_1\!\wedge\!F_2$\\
\scriptsize{26}\>    ---  exit \prd ---\\[4pt]
\scriptsize{27}\> $K_2$ is a global certif.\\
\scriptsize{28}\> $F_1 := F_1 \setminus \s{C_1}$ \\
\scriptsize{29}\> Sol. $F^*_1=F_1 = \s{C_4}$ \\
\end{tabbing} 
\vspace{-10pt}
%\caption{Operation of \Apqe}
%\caption{An example of how \Apqe operates}
\caption{\Apqe, an example of operation}
\vspace{3pt}
\label{fig:exmp1}
%\end{figure}
\end{wrapfigure}

First, \Apqe picks $C_1$, the only quantified clause of $F_1$.  We
will refer to $C_1$ as the \tb{target clause}. Then \Apqe invokes a
procedure called \prd to prove $C_1$ redundant (lines 8-26). The
algorithm of \prd is somewhat similar to that of a
SAT-solver~\cite{grasp}. \prd makes decision assignments and runs
\ti{BCP} (Boolean Constraint Propagation). Besides, \prd uses the
notion of a \tb{decision level} that consists of a decision assignment
and implied assignments derived by \ti{BCP}. (The decision level
number 0 is an exception. It has only implied assignments.) On the
other hand, there are a few important differences.  In particular,
\prd has a richer set of backtracking conditions, a conflict being
just one of them.

\prd starts the decision level number 1 by making assignment $y_1 =
0$.  Then it runs \ti{BCP} to derive assignments $x_2=1$ and $x_3=0$
from clauses $C_2$ and $C_3$ that became \tb{unit} (i.e. have only one
unassigned variable). At this point, a conflict occurs since $C_1$ is
falsified (lines 11-16). Then \prd generates conflict clause $C_4 =
y_1$.  It is built like a regular conflict clause~\cite{grasp}.
Namely, $C_4$ is obtained by resolving $C_1$ with $C_2$ and $C_3$ to
eliminate the variables whose values were derived by \ti{BCP} at
decision level 1.  The clause $C_4$ \ti{certifies} that $C_1$ is
redundant in \prob{X}{F_1 \wedge F_2} in subspace $y_1=0$. We call a
clause like $C_4$ \tb{a certificate}.  Note that $C_1$ becomes
redundant only after \ti{adding} $C_4$ to the formula, because $C_1$
itself is involved in the derivation of $C_4$.  We will refer to the
certificates one has to add to the formula as \tb{participant
  certificates}. The participant certificates depending only on free
variables form a solution to the PQE problem.

After generating $C_4$, like a SAT-solver, \prd backtracks to the
smallest decision level where $C_4$ is unit (i.e. level 0) and derives
the assignment $y_1\!=\!1$. Then the target $C_1$ is blocked at
variable $x_2$ (lines 19-25). The reason is that $C_2$, the only
clause resolvable with $C_1$ on $x_2$, is satisfied by $y_1\!=\!1$.
At this point, \prd generates the clause
$K_1\!=\!\overline{y}_1\!\vee\overline{x}_2$.  It implies $C_1$ in
subspace $y_1 = 1$, thus certifying its redundancy there. (The
construction of $K_1$ is explained in Example~\ref{exmp:blk_cls} of
Subsection~\ref{ssec:blk_trg}.  Importantly, the target $C_1$ \ti{is
  not used} in generation of $K_1$.)  By resolving $K_1$ and $C_4 =
y_1$, \prd builds the final certificate $K_2 = \overline{x}_2$ for the
decision level 0.  \prd derives $K_2$ from $K_1$ like a SAT-solver
derives a conflict clause from a clause falsified at a conflict
level. That is $K_2$ is built by resolving out variables of $K_1$
assigned by values derived at the current decision level. In our case,
it is the variable $y_1$. Since $K_1$ and $K_2$ are derived without
using the target clause $C_1$, one \ti{does not have to} add them to
the formula. They just ``witness'' the redundancy of $C_1$.  We will
refer to them as \tb{witness certificates}.

$K_2$ implies $C_1$ in the entire space and thus is a global
certificate. So, \Apqe removes $C_1$ from $F_1$ (line 28). Since now
$F_1$ does not have quantified clauses, \Apqe terminates. It returns
the current $F_1= \s{C_4}$ as a solution $F^*_1(Y)$ to the PQE
problem.  That is \prob{X}{C_1 \wedge F_2} $\equiv C_4\,\wedge$
\prob{X}{F_2}.
%\clearpage

\section{Description Of \Apqe}
\label{sec:start}

In this section, we describe \Apqe in more detail.  A proof of
correctness of \Apqe is given in Appendix~\ref{app:sound_compl}.  For
the sake of simplicity, in the current version of \Apqe, the witness
certificates are not added to the formula and so are not
reused\footnote{In practice, witness certificates are derived in
  subspaces where the formula is satisfiable. So, reusing them should
  boost the pruning power of \Apqe in those subspaces.}.

%
% Subsection: Main loop of \Apqe
%
%\vspace{-3pt}
\subsection{The main loop of \Apqe}
%
%
% \Apqe procedure
%
%\vspace{-10pt}
\setlength{\intextsep}{4pt}
\begin{wrapfigure}{l}{1.8in}
%% \begin{figure}
%\begin{center}
\centering
\small
%\normalsize
%\vspace{5pt}
\parbox{0cm}{\begin{tabbing}
%\begin{tabbing}
%aa\=bb\= cc\= ddddddd\= \kill
a\=b\=cc\= ddddddd\= \kill
%// $\lambda$ stands for $\!\Sub{C}{trg},\!Y,\!\pnt{q}$ \\
%// \\
$\Apqe(F_1,F_2,Y)$\{\\
\scriptsize{1}\>\,while $(\mi{true})$ \{ \\
\scriptsize{2}\Tt  $\Sub{C}{trg}:= \mi{PickQntCls}(F_1)$   \\
\scriptsize{3}\Tt  if $(\Sub{C}{trg} = \mi{nil})$ \{ \\
\scriptsize{4}\ttt   $F^*_1 := F_1$ \\
\scriptsize{5}\ttt return($F^*_1$)\}\\
\scriptsize{6}\Tt  $\pnt{q} := \emptyset$ \\
\scriptsize{7}\Tt  $K\!:=\!\mi{PrvRed}(F_1\!\wedge\!F_2,\!\Sub{C}{trg},\!Y,\!\pnt{q})$\\
\scriptsize{8}\Tt  if $(\mi{EmptyCls}(K))$ return($K$) \\
\scriptsize{9}\Tt  $F_1 := F_1 \setminus$ \s{\Sub{C}{trg}}\}\} \\
%% \tb{\scriptsize{1}}\>   \\
\end{tabbing}}
\vspace{-15pt}
\caption{The main loop}
%\vspace{-15pt}
\label{fig:mloop}
%\end{center}
%% \end{figure}
\end{wrapfigure}

The main loop of \Apqe is shown in Fig.~\ref{fig:mloop}. \Apqe accepts
formulas $F_1(X,Y)$, $F_2(X,Y)$ and set $Y$ and outputs formula
$F^*_1(Y)$ such that $\prob{X}{F_1 \wedge F_2} \equiv F^*_1 \wedge
\prob{X}{F_2}$. The loop begins with picking a quantified clause
$\Sub{C}{trg} \in F_1$ that is the \ti{target} clause to be proved
redundant (line 2). If $F_1$ has no quantified clauses, it is the
solution $F^*_1(Y)$ returned by \Apqe (lines 3-5). Otherwise, \Apqe
initializes the assignment \pnt{q} to $X \cup Y$ and invokes a
procedure called \prd to prove \Sub{C}{trg} redundant (lines 6-7).
\prd returns a clause $K$ implying \Sub{C}{trg} and thus
\ti{certifying} its redundancy. If $K$ is an empty clause (i.e. has no
literals), $F$ is unsatisfiable. Then \prd returns $K$ as a solution
to the PQE problem (line~8).  Otherwise, $K$ consists of (some)
literals of \Sub{C}{trg}. Besides, $K$ is redundant in \prob{X}{K
  \wedge (F_1 \cup F_2 \setminus \s{\Sub{C}{trg}})}.  So, \Sub{C}{trg}
is redundant in \prob{X}{\Sub{C}{trg} \wedge (F_1 \cup F_2 \setminus
  \s{\Sub{C}{trg}})} and \Apqe removes it from $F_1$ (line 9).
In the process of deriving the certificate $K$ above, \prd may add
participant certificates to $F_1$.  If an added certificate clause
$K'$ is quantified, \prd will be called at a later iteration of the
main loop to prove $K'$ redundant.

%
% subsection: PrvRed
%
%\vspace{-3pt}
\subsection{Description of \prd}
\label{ssec:prv_red}
%\vspace{-5pt}

The pseudo-code of \prd is shown in Fig~\ref{fig:prv_red}. Let $F$
denote $F_1\,\wedge\,F_2$. The objective of \prd is to prove the current
target clause \Sub{C}{trg} redundant in \prob{X}{F} in the subspace
specified by an assignment \pnt{q}\, to $X \cup Y$.  The reason why
one needs \pnt{q}\, is that \prd can be called recursively in
\ti{subspaces} to prove redundancy of some ``local'' target clauses
(Section~\ref{sec:trg_unit}).

First, in line 1, \prd stores the initial value of \pnt{q}. (It is
used in line 10 to limit the backtracking of \prd.) Besides, \prd
initializes the assignment queue $Q$.  The main work is done in a loop
similar to that of a SAT-solver~\cite{grasp}.  The operation of \prd
in this loop is partitioned into two parts separated by the dotted
line.

%
% PrvRed procedure
%
\setlength{\intextsep}{4pt}
\begin{wrapfigure}{l}{2in}
%\begin{figure}
\centering
%\begin{center}
\small
%\normalsize
%\vspace{-5pt}
\parbox{0cm}{\begin{tabbing}
%% aaa\=bb\= cc\= dd\= \kill
a\=bb\= cc\= dd\= \kill
// $F$ denotes $F_1 \wedge F_2$ \\
// \\
$\mi{PrvRed}(F,\Sub{C}{trg},\!Y,\!\pnt{q})$\{\\
\scriptsize{1}\>\, $\Sub{\pnt{q}}{init} := \pnt{q}$; $Q = \emptyset$ \\
\scriptsize{2}\>\, while ($\mi{true}$) \{\\
\scriptsize{3}\Tt if ($Q = \emptyset$) \{ \\
\scriptsize{4}\ttt  $(v,b)\!:=\!\mi{MakeDec}(F,\!Y,\!\Sub{C}{trg})$\\
\scriptsize{5}\ttt  $\mi{UpdQueue}(Q,v,b)$ \} \\
\scriptsize{6}\Tt $\Sub{K}{bct}\!:=\!\mi{BCP}(Q,\!\pnt{q},F,\!Y,\!\Sub{C}{trg})$\\
\scriptsize{7}\Tt if $(\Sub{K}{bct} = \mi{nil})$ continue \\
$~~~-----$\\
\scriptsize{8}\Tt $K\!:=\!\mi{Lrn}(F,\pnt{q},\!\Sub{K}{bct})$\\
\scriptsize{9}\Tt if $(\mi{Particip}(K))$ $F_1\!:=\!F_1\!\cup\!\s{K}$  \\
\scriptsize{10}\Tt  $\mi{Backtrack}(\Sub{\pnt{q}}{init},\pnt{q},K)$ \\
\scriptsize{11}\Tt if $(\pnt{q} = \Sub{\pnt{q}}{init})$ return($K$) \\
\scriptsize{12}\Tt $\mi{UpdQueue}(Q,\pnt{q},K)$\}\} \\
%\tb{\scriptsize{1}} \\
%\tb{\scriptsize{1}} \\
\end{tabbing}}
\vspace{-10pt}
\caption{The \prd procedure}
%\vspace{-5pt}
\label{fig:prv_red}
%\end{figure}
\end{wrapfigure}

The first part (lines 3-7) starts with checking if the assignment
queue $Q$ is empty. If so, a decision assignment $v\!=\!b$ is picked
and added to $Q$ (lines 4-5). Here $v\!\in\!(X \cup Y)$ and $b\!\in\!
\s{0,1}$. The variables of $Y$ are the first to be assigned by
\prd\footnote{\label{ftn:dec_order}
The goal of \Apqe is to derive free clauses making the quantified
clauses of $F_1$ redundant in \prob{X}{F_1 \wedge F_2}. Assigning
variables of $X$ after those of $Y$ guarantees that, when generating a
new clause, the variables of $X$ are resolved out \ti{before} those of
$Y$.

}. So \mbox{$v\!\in\!X$}, only if all
variables of $Y$ are assigned. If $v \in \V{\Sub{C}{trg}}$, then $v=b$
is picked so as to \ti{falsify} the corresponding literal of
\Sub{C}{trg}.  (\Sub{C}{trg} is obviously redundant in subspaces where
it is satisfied.)

Then \prd calls the
\ti{BCP} procedure.  If \ti{BCP} identifies a backtracking condition,
it returns a certificate clause \Sub{K}{bct} implying \Sub{C}{trg} in
the current subspace.  (Here, ``\ti{bct}'' stands for ``backtracking''
because \Sub{K}{bct} is the reason for backtracking.)  After \ti{BCP},
\prd goes to the second part of the loop where the actual backtracking
is done. If no backtracking condition is met, a new iteration begins.

The certificate \Sub{K}{bct} returned by \ti{BCP} depends on the
backtracking condition. \ti{BCP} identifies three of them: a) a
conflict, b) \Sub{C}{trg} is implied in subspace \pnt{q}\,\,by an
existing clause, and c) \Sub{C}{trg} is blocked in subspace
\pnt{q}. In the first case, \Sub{K}{bct} is a clause falsified in the
current subspace \pnt{q} i.e. one reached during BCP. In the second
case, \Sub{K}{bct} is a clause that \ti{BCP} made unit and that shares
its only literal with \Sub{C}{trg}. (Such a clause implies
\Sub{C}{trg} in the current subspace \pnt{q}.)  In the third case,
\Sub{K}{bct} is generated by \prd as described in the next subsection.

\prd starts the second part (lines 8-12) with a procedure called \lrn
that uses \Sub{K}{bct} to build another certificate $K$ implying
\Sub{C}{trg} in subspace \pnt{q}. Generation of $K$ from \Sub{K}{bct}
is similar to how a SAT-solver generates a conflict clause from a
falsified clause~\cite{grasp}.  Namely, when building $K$, \lrn
resolves out the variables whose value was \ti{derived} at the
decision level where the backtracking condition occurred.
If \Sub{C}{trg} was used to generate $K$ i.e. the latter is a
\ti{participant certificate}, $K$ is added to $F_1$ (line 9).  This
guarantees that \prd adds only clauses \ti{implied} by the current
formula.  (The only es-clauses generated by \prd and described in the
next subsection are used solely to generate \ti{witness certificates}.
So, a witness certificate $K$ is, in general, es-implied rather than
implied by the formula $F$ in \prob{X}{K \wedge F}. For that reason,
in the current version of \Apqe, witness certificates are not added to
the formula. In one special case, to avoid adding a witness
certificate, \prd has to derive and add to the formula a special
clause. This case is described in Appendix~\ref{app:spec_case}.)

After generating $K$, \prd backtracks (line 10). The assignment
\Sub{\pnt{q}}{init} sets the limit of backtracking.  If \prd reaches
this limit, \Sub{C}{trg} is proved redundant in the required subspace
and \prd terminates (line 11). Otherwise, an assignment is derived
from $K$ and added to the queue $Q$ (line 12). This is similar to the
backtracking of a SAT-solver to the smallest decision level where the
last conflict clause is unit.  So, an assignment can be derived from
this clause by \ti{BCP}. More information can be found in
Appendix~\ref{app:bcktr}.

%
% Subsection: adding an es-clause
%
%\vspace{-5pt}
%\subsection{Generation of a clause implying \Sub{C}{trg} when \Sub{C}{trg} is blocked}
\subsection{Generation of clause \Sub{K}{bct} when \Sub{C}{trg} is blocked}

%\subsection{Generation of a certificate when \Sub{C}{trg} is blocked}
\label{ssec:blk_trg}
Let \Sub{C}{trg} get blocked in \prob{X}{F} in the current subspace
\pnt{q} during \ti{BCP}.  So, \Sub{C}{trg} is redundant in \prob{X}{F}
in this subspace. Then a clause \Sub{K}{bct} is generated as described
in Proposition~\ref{prop:blk_to_es} where
\mbox{$\cof{(\Sub{K}{bct})}{q} \imp \cof{(\Sub{C}{trg})}{q}$}\, and
\Sub{K}{bct} is redundant in \prob{X}{\Sub{K}{bct} \wedge (F \setminus
  \s{\Sub{C}{trg}})}. Thus, \Sub{K}{bct} certifies redundancy of
\Sub{C}{trg} in subspace \pnt{q} and is returned by \ti{BCP} as the
reason for backtracking (line 6 of Fig~\ref{fig:prv_red}). This is the
only case of backtracking where the clause \Sub{K}{bct} returned by
BCP is \ti{es-implied} rather than implied by $F$ in
\prob{X}{\Sub{K}{bct} \wedge F}.

%
% Proposition: generating an es-implied clause from a blocked one
%
\begin{proposition}
  \label{prop:blk_to_es}
Given a formula \prob{X}{F(X,Y)}, let $\Sub{C}{trg} \in F$.  Let
\pnt{q} be an assignment to $X \cup Y$ that does not satisfy
\Sub{C}{trg}. Let \Sub{C}{trg} be blocked in \prob{X}{F} at $w \in X$
in subspace \pnt{q} where \mbox{$w \not\in \Va{q}$}. Let $l$($w$) be
the literal of $w$ present in \Sub{C}{trg}.  Let $K'$ denote the
longest clause falsified by \pnt{q}. Let $K''$ be a clause formed from
$l(w)$ and a subset of literals of \Sub{C}{trg} such that every clause
of \cof{F}{q} unresolvable with \cof{(\Sub{C}{trg})}{q} on $w$ is
unresolvable with \cof{(K'')}{q} too. Let $\Sub{K}{bct} = K' \vee
K''$. Then \mbox{$\cof{(\Sub{K}{bct})}{q} \imp
  \cof{(\Sub{C}{trg})}{q}$} and \Sub{K}{bct} is redundant in
\prob{X}{\Sub{K}{bct} \wedge (F \setminus \s{\Sub{C}{trg}})}.
\end{proposition}

%\vspace{-5pt}
\begin{example}
\label{exmp:blk_cls}
Let us recall the example of Section~\ref{sec:exmp1}.  Here we have a
formula \prob{X}{F} where $X=\s{x_2,x_3}$, $Y=\s{y_1}$, $F = C_1
\wedge C_2 \wedge C_3 \wedge C_4$, $C_1 = \overline{x}_2 \vee x_3$,
$C_2 = y_1 \vee x_2$, $C_3 = y_1 \vee \overline{x}_3$, $C_4 = y_1$. In
subspace $y_1=1$, the target clause $C_1$ is \ti{blocked} at $x_2$ and
hence is redundant. ($C_1$ can be resolved on $x_2$ only with $C_2$
that is satisfied by $y_1=1$.) This redundancy can be certified by the
clause $K_1 = \overline{y_1} \vee \overline{x}_2$ implying $C_1$ in
subspace $y_1=1$. The clause $K_1$ is constructed as $K' \vee K''$ of
Proposition~\ref{prop:blk_to_es}.  Here $K' = \overline{y}_1$ is the
clause falsified by the assignment $y_1=1$. The clause
$K''=\overline{x}_2$ has the same literal of the blocked variable
$x_2$ as the target clause $C_1$.  (Formula $F$ has no clauses
unresolvable with $C_1$ on $x_2$. So, $K''$ needs no more literals.)
$\blacksquare$
\end{example}
% 
% remark: blocked target clause is unit
%
%\vspace{-5pt}
\begin{remark}
\label{rem:blk_to_es}
Let \Sub{C}{trg} of Proposition~\ref{prop:blk_to_es} be \ti{unit} in
subspace \pnt{q} (and $w$ be the only unassigned variable of
\Sub{C}{trg}). Then $K''$ reduces to $l(w)$ and $\Sub{K}{bct} = K'
\vee \l(w)$.
\end{remark}

\section{The Case When The Target Clause Becomes Unit}
\label{sec:trg_unit}
%\input{s6imp_examp.fig}
%\vspace{4pt}
In this section, we describe what \prd does when the current target
clause \Sub{C}{trg} becomes unit. (Since \prd first assigns variables
of $Y$, the unassigned variable of \Sub{C}{trg} is in $X$
i.e. \ti{quantified}.)  In this case, \prd recursively calls itself to
prove redundancy of every clause resolvable with \Sub{C}{trg}. A
concrete example is given in Appendix~\ref{app:exmp}.

Figure~\ref{fig:bcp} shows the fragment of \ti{BCP} invoked when the
current target \Sub{C}{trg} becomes unit. Let $x \in X$ denote the
only unassigned variable of \Sub{C}{trg}. Assume for the sake of
clarity that \Sub{C}{trg} contains the positive literal of $x$. At
this point a SAT-solver would derive the assignment $x=1$ because
\Sub{C}{trg} is falsified under assignment $x=0$.  However, the goal
of \prd is to prove \Sub{C}{trg} redundant rather than find a
satisfying assignment. The fact that \Sub{C}{trg} is falsified in a
subspace says nothing about whether it is redundant there.

%
% BCP procedure
%
\setlength{\intextsep}{4pt}
\setlength{\textfloatsep}{10pt}
\begin{wrapfigure}{l}{2in}
%\begin{figure}[h!]
\centering
%\begin{center}
\small
%\normalsize
%\vspace{5pt}
\parbox{0cm}{\begin{tabbing}
aaa\=bb\=cc\= dd\= \kill
//  $F$ denotes $F_1 \wedge F_2$ \\
// \\
$\mi{BCP}(Q,\pnt{q},F,Y,\Sub{C}{trg})$ \{ \\
~~~~~~~~~~~$\cdot\cdot\cdot$  \\
\scriptsize{10}\> if $(\mi{Unit}(\Sub{C}{trg},\pnt{q}))$ \{  \\
\scriptsize{11}\Tt   $(\Sub{K}{bct},G) := \mi{Rcrs}(F,\Sub{C}{trg},\pnt{q})$\\
\scriptsize{12}\Tt   if ($\Sub{K}{bct} \neq \mi{nil}$) return(\Sub{K}{bct})\\
\scriptsize{13}\Tt   $\Sub{K}{bct} := \mi{GenCert}(F,\Sub{C}{trg},\pnt{q},G)$\\
\scriptsize{14}\Tt   return(\Sub{K}{bct}) \} \\
~~~~~~~~~~~$\cdot\cdot\cdot$  \\
%\tb{\scriptsize{15}}\> \\
%% \tb{\scriptsize{1}}\>   \\
%% \tb{\scriptsize{1}}\>   \\
\end{tabbing}}
\vspace{-15pt}
\caption{A fragment of \ti{BCP}}
%\vspace{-5pt}
\label{fig:bcp}
%\end{figure}
\end{wrapfigure}

So, \ti{BCP} invokes procedure \ti{Rcrs} that recursively calls \prd
for every clause resolvable with \Sub{C}{trg} on $x$. The name
\ti{Rcrs} abbreviates ``recurse''.  This call can have two outcomes.
First, \ti{Rcrs} may return a clause \Sub{K}{bct} that is falsified by
\pnt{q}. (This is possible only if $F$ is unsatisfiable in subspace
\pnt{q}.)  Then \ti{BCP} returns \Sub{K}{bct} as the reason for
backtracking (line 12). Second, \ti{Rcrs} proves the clauses
resolvable with \Sub{C}{trg} on $x$ redundant and returns a set $G$ of
certificates. For each clause $C$ resolvable with \Sub{C}{trg} on $x$,
the set $G$ contains a certificate of redundancy of $C$ in subspace
$\pnt{q}\, \cup \s{x=1}$.  At this point, \Sub{C}{trg} is blocked at
$x$ in subspace \pnt{q}. So, a certificate \Sub{K}{bct} is built using
Proposition~\ref{prop:blk_to_es} (line 13). It is returned by \ti{BCP}
as the reason for backtracking.
%
% remark
%
\begin{remark}
  \label{rem:loc_rem}
Every clause $C$ resolvable with \Sub{C}{trg} on $x$ and proved
redundant in subspace $\pnt{q}\, \cup \s{x=1}$ is \ti{temporarily}
removed from the formula $F$ until backtracking. Since $C$ is proved
redundant only locally, one has to \ti{return} it to $F$ after
backtracking.  Nevertheless, \Sub{C}{trg} remains blocked in subspace
\pnt{q} and hence redundant there (see Remark~\ref{rem:blk_cls} of
Section~\ref{sec:es_impl}).
\end{remark}
%\vspace{-5pt}

%\clearpage

\section{Invariant Generation For Bug Detection}
\label{sec:inv_gen}
In this section, we discuss using PQE for bug detection by invariant
generation.  An \tb{invariant} $P$ of a sequential circuit $N$ is a
formula satisfied by every reachable state of $N$.  So, the states
falsifying $P$ are unreachable in $N$. We will call an invariant
\tb{local} if it holds in \ti{some} time frames. To distinguish
between local invariants and those holding in every time frame we will
call the latter \tb{global}. When we say ``invariant'' without a
qualifier we mean a global invariant.
%
%  Subsection: Two kinds of bugs
%
\subsection{Two kinds of bugs}
\label{ssec:two_kinds}
Let \bm{N} be a sequential circuit. Let $P_1(S)$,$\dots$,$P_n(S)$ be
invariants that must hold for $N$ where $S$ is the set of state
variables. That is, these are \ti{desired} invariants of $N$. One can
view the aggregate invariant $P_1\wedge \dots \wedge P_n$ as a
\ti{specification} \bm{\spe} for $N$. We will say that \pnt{s} is a
\tb{bad state} (respectively a \tb{good state}) if
$\mi{Sp}(\pnt{s})=0$ (respectively $\mi{Sp}(\pnt{s})=1$). As far as
reachable states are concerned, $N$ can have \tb{two kinds of bugs}. A
\ti{bug of the first kind} occurs when a bad state is reachable in
$N$. A \ti{bug of the second kind} takes place when a good state that
is supposed to be reachable is unreachable in $N$. Informally, a bug
of the first kind (respectively the second kind) indicates that the
set of reachable states is ``larger'' (respectively ``smaller'') than
it should be.

To prove that $N$ has no bugs of the first kind, it suffices to show
that the aggregate invariant \spe holds for $N$.  Note that this does
nothing to identify bugs of the \ti{second} kind. Indeed, let
\Sub{N}{triv} be a circuit looping in an initial state \ppnt{s}{init}
satisfying \spe. Then \spe holds for \Sub{N}{triv}. However,
\Sub{N}{triv} has bugs of the second kind (assuming that a correct
implementation has to reach states other than \ppnt{s}{init}). A
straightforward way to identify bugs of the second kind is to compute
the set of all unreachable states of $N$. If this set contains a state
that is supposed to be reachable, $N$ has a bug of the second
kind. Unfortunately, computing such a set can be prohibitively hard.

Note that one cannot \ti{prove} the existence of a bug of the second
kind by testing: the unreachability of a state cannot be established
by a counterexample. However, testing can point to the possibility of
such a bug (see Section~\ref{sec:be_test}).  An important method for
finding bugs of the second kind in a circuit $N$ is to identify its
unwanted invariants.  We will call $Q$ an \tb{unwanted invariant} if
it is falsified by a state \pnt{s} that is supposed to be
reachable. If $Q$ holds for $N$, then \pnt{s} is unreachable and $N$
has a bug of the second kind. Currently, unwanted invariants are
detected via checking a list of expected events~\cite{syst_veril}. (If
an event of this list never occurs, $N$ has an unwanted invariant.)
This list is formed \ti{manually}. So, in a sense, unwanted invariants
are simply guessed. For instance, one can check if $N$ reaches a state
where a state variable $s_i \in S$ changes its initial value. If not,
then $N$ has an unwanted invariant, assuming that states with both
values of $s_i$ are supposed to be reachable in $N$. (For the circuit
\Sub{N}{triv} above, this unwanted invariant holds for every state
variable.) The problem with guessing unwanted invariants is that, in
general, they are as unpredictable as bugs.

In this paper, we consider an approach to finding bugs of the second
kind where \tb{invariants are generated} automatically in a systematic
way.  The necessary condition for an invariant $Q$ to be unwanted is
$\spe \not\imp Q$. (If $\spe \imp Q$, then $Q$ is a desired invariant
of $N$.)  So, the overall idea is to generate invariants of $N$ not
implied by \spe and \ti{check if any of them is unwanted}. In some
cases, the designer can tell if $Q$ is an unwanted
invariant. Otherwise, one needs to find a bug-exposing test as
explained in Section~\ref{sec:be_test}. In general, an invariant
specifies only a subset of unreachable states of $N$.  So, it can be
generated much more efficiently than the entire set of unreachable
states.

%\clearpage
%
% Subsection: invariant generation by PQE
%
\vspace{-7pt}
\subsection{Invariant generation by PQE}
\label{ssec:pg_by_pqe}
Let us show how one can generate invariants by PQE. First, we consider
the generation of a local invariant that holds in $k$-th time frame.
So, a state falsifying such an invariant is unreachable in $k$
transitions. Then we show that a local invariant can be used to
generate global invariants.  Let formulas $I$ and $T$ specify the
initial states and the transition relation of $N$ respectively. Let
\bm{F_k} denote the formula obtained by unfolding $N$ for $k$ time
frames. That is $F_k = I(S_0) \wedge T(S_0,S_1) \wedge \dots \wedge
T(S_{k-1},S_k)$ where $S_j$ denotes the state variables of $j$-th time
frame, \mbox{$0\!\leq\!j\!\leq k$}. (For the sake of simplicity, in
$T$, we omit the combinational i.e. unlatched variables of $N$.)

Let $H_k(S_k)$ be a solution to the PQE problem of taking a clause $C$
out of \prob{\Abs{k-1}}{F_k} where $\Abs{k-1} = S_0 \cup \dots \cup
S_{k-1}$.  That is \prob{\Abs{k-1}}{F_k} $\equiv H_k \wedge$
\prob{\Abs{k-1}}{F_k \setminus \s{C}}. Since $F_k$ implies $H_k$, the
latter is a \tb{local invariant} of $N$ holding in $k$-th time frame.
Note that performing full QE on \prob{\Abs{k-1}}{F_k} produces the
\ti{strongest} local invariant specifying all states unreachable in
$k$ transitions.  Computing this invariant can be prohibitively hard.
PQE allows to build a collection of \ti{weaker} local invariants $H_k$
each specifying only a subset of states unreachable in $k$
transitions.  Computation of such invariants can be dramatically more
efficient since PQE can be much easier than QE.

One can use $H_k$ to find global invariants as follows.  The fact that
$H_k$ is not a global invariant does not mean that \ti{every clause}
of $H_k$ is not a global invariant either.  On the contrary, the
experiments presented in Appendix~\ref{app:exper2} showed that even
for small $k$, a large share of clauses of $H_k$ were a global
invariant. (To find out if a clause $Q \in H_k$ is a global invariant,
one can simply run a model checker to see if $Q$ holds.) 

\subsection{Using Invariant Generation}
\label{ssec:use_pg}
One of possible ways to use invariant generation is to take out
clauses according to some coverage metric.  The intuition here is
based on the two observations below.  Let $Q$ be an invariant obtained
by taking a clause $C$ out of \prob{\Abs{k-1}}{F_k}. The first
observation is that the states falsifying $Q$ are unreachable due to
the \ti{presence} of $C$.  So, if a part of the circuit $N$ is
responsible for a bug of the second kind and $C$ is related to this
part, taking out $C$ may produce an \ti{unwanted} invariant. This
observation is substantiated in the next section. The second
observation is that by taking out different clauses one generates
different invariants ``covering'' different parts of the circuit
$N$. An example of a coverage metric is presented in the next section.
There we take out the clauses containing an unquantified variable of
\prob{\Abs{k-1}}{F_k} (i.e. a state variable of the $k$-th time
frame).  One can view such a choice of clauses as a way to cover the
design in terms of \ti{latches}.

%\pagebreak
\section{An Experiment With FIFO Buffers}
\label{sec:exper1}
In this section, we describe an experiment with FIFO buffers. Our
objective here is twofold. First, we explain how bug detection by
invariant generation works on a practical example. Second, we want to
show that even the current version of \Apqe whose performance can be
dramatically improved can address an important practical problem.  In
Appendix~\ref{app:exper2}, we apply \Apqe to invariant generation for
HWMCC-13 benchmarks. We also use these benchmarks to compare PQE with
QE and \Apqe with \apqe, our previous PQE-solver~\cite{hvc-14}. In
this section, such a comparison is done on FIFO buffers. A Linux
binary of \Apqe and a sample of formulas used in the experiments can
be downloaded from~\cite{start_binary}. In all experiments, we used a
computer with Intel Core i5-8265U CPU of 1.6\,GHz.

%
% subsection
%
\vspace{-5pt}
\subsection{Buffer description}
%
% verilog code
%
%\setlength{\intextsep}{4pt}
%\setlength{\textfloatsep}{10pt}
\begin{wrapfigure}{l}{1.9in}
%\begin{figure}[h!]
\centering
%\begin{center}
\small
%\normalsize
\vspace{-10pt}
\parbox{0cm}{\begin{tabbing}
aa\=bb\=cc\= dd\= \kill
~~~~~~~~~~~$\cdot\cdot\cdot$  \\
if ($\mi{write}==1~\&\&~\mi{currSize}< n$) \\
*\>if ($\mi{dataIn}~!\!= \mi{Val}$) \\
\Tt begin \\
\Tt     $\data[\mi{wrPnt}]  = \mi{dataIn}$; \\
\Tt     $\mi{wrPnt}  = \mi{wrPnt}+1$; \\
\Tt   end \\
~~~~~~~~~~~$\cdot\cdot\cdot$  \\
%\tb{\scriptsize{15}}\> \\
%% \tb{\scriptsize{1}}\>   \\
%% \tb{\scriptsize{1}}\>   \\
\end{tabbing}}
\vspace{-25pt}
\caption{A buggy fragment of Verilog code describing \fifo}
\vspace{5pt}
\label{fig:bug}
%\end{figure}
\end{wrapfigure}

\vspace{-5pt}

In this section, we give an example of bug detection by invariant
generation for a FIFO buffer called \fifo.  Let $n$ be the number of
elements of \fifo and \data denote the data buffer of \fifo.  Let each
$\data[i],i=1,\dots,n$ have $p$ bits and be an integer where $0 \leq
\data[i] < 2^p$.  A fragment of the Verilog code describing \fifo is
shown in Fig~\ref{fig:bug}. This fragment has a buggy line marked with
an asterisk. In the correct version without the marked line, a new
element $\mi{dataIn}$ is added to \data if the \ti{write} flag is on
and \fifo holds less than $n$ elements.  Since \data can have any
combination of numbers, all \data states are supposed to be reachable.

However, due to the bug, the number $\mi{Val}$ cannot appear in \data.
(Here $\mi{Val}$ is some constant \mbox{$0\!<\!\mi{Val}\!<\!
  2^p$}. We assume that the buffer elements are initialized to
0.) So, \fifo has a \ti{bug of the second kind} since it cannot reach
states where an element of \data equals $\mi{Val}$.  This bug is hard
to detect by random testing because it is exposed only if one tries to
add \ti{Val} to \fifo.  Similarly, it is virtually impossible to guess
an unwanted invariant of \fifo exposing this bug unless one knows
exactly what this bug is.

%
% subsection
%
\vspace{-5pt}
\subsection{Bug detection by invariant generation}
Let $N$ be a circuit implementing \fifo. Let $S$ be the set of state
variables of $N$ and \bm{\Sub{S}{data} \subset S} be the subset
corresponding to the data buffer \data.  We used \Apqe to generate
invariants of $N$ as described in the previous section.  Note that an
invariant $Q$ depending only on \Sub{S}{data} is an unwanted one.  If
$Q$ holds for $N$, some states of \data are unreachable. Then \fifo
has a \ti{bug of the second kind} since every state of \data is
supposed to be reachable. To generate invariants, we used the formula
$F_k = I(S_0) \wedge T(S_0,S_1) \wedge \dots \wedge T(S_{k-1},S_k)$
introduced in Subsection~\ref{ssec:pg_by_pqe}. Here $I$ and $T$
describe the initial states and the transition relation of $N$
respectively and $S_j$ is the set of state variables in $j$-th time
frame. First, we used \Apqe to generate local invariants
$H_k$. Namely, $H_k$ was obtained by taking a clause $C$ out of
\prob{\Abs{k-1}}{F_k} where $\Abs{k-1} = S_0 \cup \dots \cup S_{k-1}$.
That is, \prob{\Abs{k-1}}{F_k} $\equiv H_k \wedge$
\prob{\Abs{k-1}}{F_k \setminus \s{C}}. We picked clauses to take out
as described in Subsection~\ref{ssec:use_pg}.  Namely, we took out
only clauses containing an unquantified variable (i.e. a state
variable of the $k$-th time frame). The time limit for solving the PQE
problem of taking out a clause was set to 10 sec.

For each clause $Q$ of every local invariant $H_k$ generated by PQE,
we checked if $Q$ was a global invariant. Namely, we used a publicly
available version of \ict~\cite{ic3,ic3_impl} to verify if the
invariant $Q$ held. If so, and $Q$ depended only on variables of
\Sub{S}{data}, $N$ had an \ti{unwanted invariant}. Then we stopped
invariant generation. The results of the experiment are given in
Table~\ref{tbl:buff}.  (In the experiment, we considered buffers with
32-bit elements.) Let us use the first line of Table~\ref{tbl:buff} to
explain its structure. The first two columns show the size of \fifo
implemented by $N$ and the number of latches in $N$ (8 and 300).  The
third column gives the number $k$ of time frames (i.e. 5).  The value
13 shown in the fourth column is the number of clauses taken out of
\prob{\Abs{k-1}}{F_k} before an unwanted invariant was generated.
That is, 13 was the number of PQE problems for \Apqe to solve.

\begin{wraptable}{l}{2.6in}
%\begin{table}[h]
\centering
%  \small
\vspace{2pt}
\scriptsize
\caption{\small{FIFO buffer with $n$ elements of 32 bits. Time
    limit is 10 sec. per PQE problem}}
%\caption{\small{Buffer with $n$ elements of $p$ bits}}
\vspace{-5pt}
%\begin{center}
%\renewcommand{\arraystretch}{1.2} % Default value: 1
%\begni{tabular}{|p{22pt}|p{36pt}|c|c|c|c|} \hline
\begin{tabular}{|p{18pt}|p{18pt}|p{20pt}|p{20pt}|p{20pt}|p{12pt}|p{12pt}|p{16pt}|} \hline
  buff.     & lat-  &  time   &clau-   &\multicolumn{3}{c|}{local single}& tot. \\ 
  size      & ches  &  fra-   & ses    &\multicolumn{3}{c|}{clause invariants} & run \\ \cline{5-7}
  ($n$)     &          &  mes    & taken  & gen.    & \multicolumn{2}{c|}{global?} & time  \\ \cline{6-7}
  &       &         & out    & invar.   & no  &  yes   & (s.)  \\ \hline
  ~~8    & ~300      &~~ 5  & ~13      &  ~~10      & ~8    & ~2   & 25      \\ \hline
  ~~8    & ~300      &~~ 10  & ~11     &  ~~4       & ~1    & ~3    & 54       \\ \hline
  ~~16    & ~560      &~~ 5  & ~26       &  ~~18       & ~16   & ~2   & 43      \\ \hline
  ~~16   & ~560      &~~ 10  & ~17      &  ~~2        & ~0    & ~2   & 78     \\ \hline  
\end{tabular}                
%\end{center}
\vspace{5pt}
\label{tbl:buff}
%\end{table}
\end{wraptable}

Let $C$ be a clause taken out of the scope of quantifiers by
\Apqe. Every free clause $Q$ generated when taking out $C$ was stored
as a local \ti{single-clause invariant}.  The fifth column shows that
when solving the 13 PQE problems above, \Apqe generated 10 free
clauses forming 10 local single-clause invariants. These invariants
held in $k$-th time frame (where \mbox{$k\!=\!5$}). The next two
columns show how many invariants out of 10 \ict proved false or true
\ti{globally} (8 and 2). The last column gives the total run time (25
sec).

For all four instances of \fifo listed in Table~\ref{tbl:buff}, the
invariants generated by \Apqe had one asserting that \fifo cannot
reach a state where an element of \data equals $\mi{Val}$. This
invariant was produced when taking out a clause of $F_k$ related to
the buggy line of Fig.~\ref{fig:bug} (confirming the intuition of
Subsection~\ref{ssec:use_pg}.) When picking a clause to take out,
i.e. a clause containing a state variable of $k$-th time frame, one
could make a good choice by pure luck.  To address this issue, we
picked clauses to take out \ti{randomly} and performed 10 different
runs of invariant generation. For each line of Table~\ref{tbl:buff},
the columns four to eight actually describe the average value of 10
runs.
%\clearpage

\subsection{Comparing PQE and QE}
%\vspace{2pt}
%\noindent{\tb{Comparing PQE and QE.}
To contrast PQE and QE, we used a
  high-quality tool \cad~\cite{cadet_qe,cadet_imp} to perform QE on
  formulas \prob{\Abs{k-1}}{F_k}.  That is, instead of taking a clause
  out of \prob{\Abs{k-1}}{F_k} by PQE, we applied \cad to perform full
  QE on this formula. As mentioned in Subsection~\ref{ssec:pg_by_pqe},
  performing QE on \prob{\Abs{k-1}}{F_k} produces the strongest local
  invariant specifying all states unreachable in $k$ transitions.
\cad failed to finish QE on \prob{\Abs{k-1}}{F_k} with the time limit
of 600 sec. On the other hand, \Apqe finished 63\% of the PQE problems
of taking a clause out of \prob{\Abs{k-1}}{F_k} in the time limit
(i.e. under 10 sec). This shows that PQE can be dramatically more
efficient than QE if only a small part of the formula gets
unquantified.

\subsection{\Apqe versus \apqe}
\label{ssec:comparison}
%% \vspace{2pt}
%% \noindent{\tb{\bm{\Apqe} versus \bm{\apqe}.}
%
  We repeated the
  experiment above using \apqe instead of \Apqe. \apqe is our previous
  PQE-solver~\cite{hvc-14} based on the machinery of
  D-sequents~\cite{fmcad12,fmcad13}. \apqe solved only 2\% of the PQE
  problems in the time limit of 10 sec. (as opposed to 63\% by \Apqe)
  and failed to generate an unwanted invariant.

\section{Identifying Unwanted Invariants}
\label{sec:be_test}
%\vspace{-10pt}
Sometimes it is easy to see that an invariant $Q$ is unwanted (e.g.
an invariant of \fifo depending only on variables of \Sub{S}{data} is
obviously unwanted). However, in general, to show that $Q$ is
unwanted, one needs to find a \tb{bug-exposing} test (or \tb{be-test}
for short.) Let \pnt{t} denote a test $(\ppnt{s}{0},\ppnt{x}{0},\dots,
\ppnt{x}{k-1})$ for a circuit $N$.  Here \ppnt{s}{0} is an initial
state of $N$ and \ppnt{x}{i}, \mbox{$0\!\leq i \!<\! k$} is a full
assignment to the combinational input variables of $N$ in $i$-th time
frame.  (Recall that so far, for the sake of simplicity, we omitted
combinational variables in the description of $N$.)

Let $(\ppnt{s}{0},\ppnt{s}{1},\dots, \ppnt{s}{k})$ be the trace
produced by the test \pnt{t} above (i.e. $N$ moves from
state \ppnt{s}{i} to
\ppnt{s}{i+1} under input \ppnt{x}{i}). We will say that \pnt{t} is a
\tb{be-test} for an invariant $Q$ if it is a \ti{counterexample} for
$Q$ in a \ti{correct} \smallskip version $N'$ of $N$.  That is  \pnt{t}
produces a trace $(\ppnt{s}{0},\pqnt{s'}{1},\dots, \pqnt{s'}{k})$ in
$N'$ where \pqnt{s'}{k} \ti{falsifies} $Q$. Consider, for instance,
the invariant $Q$ stating that \fifo cannot have the number $\mi{Val}$
in $j$-th element of its data buffer.  Let $\pnt{t} =
(\ppnt{s}{0},\ppnt{x}{0},\dots, \ppnt{x}{k-1})$ be a test such that
when applied to a correct design, \pnt{t} would make $\mi{Val}$ appear
in the $j$-th element of the data buffer. Then \pnt{t} is a be-test
for $Q$.

Finding a be-test is based on the following idea. Let an invariant
clause $Q$ be extracted from a formula $H_k$ obtained by taking a
clause $C$ out of \prob{\Abs{k-1}}{F_k} as described above. As we
mentioned in Remark~\ref{rem:noise}, \prob{\Abs{k-1}}{F_k} $\equiv H_k
\wedge$ \prob{\Abs{k-1}}{F_k \setminus \s{C}} holds even if the
clauses implied by $F_k \setminus \s{C}$ are removed from $H_k$. So,
we will assume that $F_k\setminus \s{C} \not\imp Q$. Then there is an
assignment \pnt{p} satisfying $(F_k \setminus \s{C})\wedge
\overline{Q}$. \smallskip One can view \pnt{p} as an execution trace
of $N$ when $C$ is removed from $F_k$.

Let
$\pnt{t^*}\!=\!(\pqnt{s^*}{0},\!\pqnt{x^*}{0},\!\dots,\!\pqnt{x^*}{k-1})$
be the test where the variables are assigned as in \pnt{p}. One can
make two claims about \pnt{t^*}. First, if $Q$ is an unwanted
invariant, \pnt{t^*} can be \ti{very close to a be-test}. Second, if
$Q$ is a \ti{desired} invariant, \pnt{t^*} is a \ti{high-quality test}
for $N$ that can be used e.g. in regression testing. The first claim
is due to \pnt{t^*} being extracted from \pnt{p} falsifying $Q$ and
satisfying all clauses of $F_k$ but $C$. The second claim is due to
\pnt{t^*} being able to detect modifications of $N$ breaking $Q$.  One
can try to produce a be-test from \pnt{t^*} either ``manually'' or
automatically generating small variations of \pnt{t^*}.

Consider, for example, the unwanted invariant $Q$ stating that the
number \ti{Val} cannot appear in $j$-th element of the data
buffer \smallskip of \fifo.  For every example of
Table~\ref{tbl:buff}, we \smallskip built the test \mbox{\pnt{t^*}=
  $(\pqnt{s^*}{0},\pqnt{x^*}{0},\dots, \pqnt{x^*}{k-1})$} extracted
from \pnt{p} satisfying $(F_k \setminus \s{C})\wedge \overline{Q}$. In
every case, \pnt{t^*} turned out to be different from a be-test only
\ti{in one bit}.

%   Plan
%   *) BDDs
%   *) SAT
%   *) Symmetry, circuits and pre-processing
%\pagebreak
\section{Some Background}
\label{sec:bg}
In this section, we discuss some research relevant to PQE and
invariant generation.  Information on BDD and SAT based QE can be
found in~\cite{bryant_bdds1,bdds_qe}
and\linebreak\cite{blocking_clause,fabio,cofactoring,cav09,cav11,cmu,nik1,nik2,cadet_qe}
respectively.
Making clauses of a formula redundant by adding resolvents is
routinely used in pre-processing~\cite{prepr,blocked_qbf} and
in-processing \cite{inproc} phases of QBF/SAT-solving. Identification
and removal of blocked clauses is also an important part of formula
simplification~\cite{tacas_blocked_clauses}. The difference of our
approach from these techniques is twofold. First, our approach employs
redundancy based \ti{reasoning} rather than formula optimization. So,
for instance, to make a target clause redundant, \Apqe can add a lot
of new clauses making the formula \ti{larger}. Second, these
techniques try to identify \ti{non-trivial} conditions under which a
clause $C$ is redundant in the \ti{entire space}.  In our approach,
one \ti{branches} to reach a subspace where proving $C$ redundant is
\ti{trivial}.  Proving redundancy of $C$ in the entire space is
achieved by merging the results of different branches.

The predecessor of the approach based on certificate clauses is the
machinery of dependency sequents
(D-sequents)~\cite{fmcad12,fmcad13}. A D-sequent is a record stating
redundancy of a clause in a quantified formula. A potential flaw of
this machinery is that to reuse a learned D-sequent, one has to keep
some contextual information~\cite{qe_learn}, which can make D-sequent
reusing expensive.  On the other hand, the reuse of certificate
clauses does not require to store any contextual information.

To the best of our knowledge, the existing procedures generate only
particular classes of invariants. For instance, they generate
invariants relating internal points of circuits to check for
equivalence ~\cite{ec_invars} or loop
invariants~\cite{loop_invars}. Another example of special invariants
are clauses generated by \ict to make an invariant $P$
inductive~\cite{ic3}. The problem here is that the closer $P$ to an
inductive invariant, the fewer invariant clauses \ict generates to
make $P$ inductive. For instance, for the circuit \Sub{N}{triv}
mentioned in Subsection~\ref{ssec:two_kinds} that loops in an initial
state, every true desired invariant $P_i$ is \ti{already
  inductive}. Hence, \ict will not generate any new invariant clauses
and will not produce an unwanted invariant even though \Sub{N}{triv}
is obviously buggy.  In Appendix~\ref{ssec:exper3}, we experimentally
compare invariants generated by \ict and \Apqe.

%\vspace{-10pt}
\section{Conclusions}
\label{sec:concl}
%\vspace{-10pt}
We consider \ti{partial} quantifier elimination (PQE) on propositional
CNF formulas with existential quantifiers. PQE allows to unquantify a
\ti{part} of the formula.  We present a PQE algorithm called \Apqe
employing redundancy based reasoning via the machinery of certificate
clauses. To prove a target clause $C$ redundant, \Apqe derives a
clause implying $C$, thus ``certifying'' its redundancy. The version
of \Apqe we describe here can still be drastically improved. We show
that PQE can be used to generate invariants of a sequential
circuit. The goal of invariant generation is to find an \ti{unwanted}
invariant of this circuit indicating that the latter is buggy. Bugs
causing unwanted invariants can be easily overlooked by the existing
methods. We applied \Apqe to identify a bug in a FIFO buffer by
generating an unwanted invariant of this buffer. We also showed that
even the current version of \Apqe is good enough to generate
invariants for HWMCC-13 benchmarks. Our experiments suggest that \Apqe
can be used for detecting hard-to-find bugs in real-life designs.

%
%  List of things to mention
%   comparison of ds-pqe and start
%   reusing
%   looping
%   add references to new papers
%\section{Relating Certificate Clauses To Some New Results}
\section{Discussing Some New Results}
\label{sec:retrosp}
Some new results have been obtained after publishing this technical
report. In this section we discuss new developments and their relation
to the machinery of certificate clauses. In
Subsection~\ref{ssec:reuse}, we describe some limitations on reusing
witness certificates. The problem of generation of duplicate
participant certificates is discussed in Subsection~\ref{ssec:dupl}.
Subsection~\ref{ssec:triv} presents some clarification on comparing
\Apqe with \apqe.  In Subsection~\ref{ssec:concl}, we give a few
concluding thoughts.

%
%  subsection
%
\subsection{Limitations on reusing certificates}
\label{ssec:reuse}
In this subsection, we briefly discuss some limitations on reusing
witness\footnote{As we mentioned in Sections~\ref{sec:exmp1} and~\ref{sec:start}, \Apqe
derives participant and witness certificates. Participant
certificates \ti{must be added} to the formula to make the current
target clause redundant. Since participant certificates are
essentially conflict clauses they can be safely reused. On the other
hand, adding witness certificates is \ti{optional} and their reusing
is not trivial. Importantly, the overwhelming majority of certificates
derived by \Apqe are witness certificates.
} certificates.  Let $K$ be a
witness certificate stating redundancy of a clause $C$ in \prob{X}{F}
in subspace \pnt{q}. As we mentioned earlier, the advantage of
deriving $K$ instead of a D-sequent is that the former can be added to
\prob{X}{F} (because $\prob{X}{F} \equiv \prob{X}{F \wedge K}$).
Adding $K$ to \prob{X}{F} makes redundancy of $C$ in subspace \pnt{q}
obvious and so enables \ti{safe} reusing of $K$ as a proof of
redundancy of $C$ in subspace \pnt{q}.

However, adding witness certificates poses the following problem. The
certificate $K$ above can be represented as $K' \vee K''$ where $K'$
is a clause falsified by \pnt{q} and clause $K''$ consists of literals
of $C$. If $K''$ contains \ti{all} the literals of $C$, then adding
$K$ to \prob{X}{F} does not make sense (because $K$ is implied by
$C$).  Of course, one can always use $K$ as a proof of redundancy in
subspace \pnt{q} \ti{without} adding it to \prob{X}{F}. However, in
this case, safe reusing of $K$ requires storing some contextual
information like it is done with D-sequents.  The theory of reusing
D-sequents (that also applies to certificate clauses) is described
in~\cite{eff_pqe}.)

%
% subsection
%
\subsection{Generation of duplicate  certificates}
\label{ssec:dupl}
When running \Apqe on multiple practical and crafted problems, we
found out that \Apqe often generates duplicate participant
certificates.  Suppose, for instance, that \Apqe proved a target
clause $C$ redundant in \prob{X}{F} in the current subspace and
temporarily removed $C$ from the formula.  When proving another target
clause $C'$, \Apqe may add to the formula a conflict clause
(certifying the redundancy of $C'$ in the current subspace) that is
identical to $C$. Since, in general, the new copy of $C$ needs to be
proved redundant too, generation of duplicate target clauses may lead
to looping\footnote{In Appendix~\ref{app:sound_compl}, to prove the
  completeness of \Apqe, we introduce its modification that does not
  loop. However, this modification is not efficient and hence
  impractical.}.

The reason for looping is that \Apqe processes one target clause at a
time and so uses a collection of search trees rather than a single
one. (On the other hand, \apqe processes all target clauses together,
which allows it to maintain a single search tree thus avoiding
looping.) There are at least two ways to prevent looping.  First, one
can just relax the requirement to process only one target clause at a
time. Suppose, for instance, that a duplicate of a target clause $C$
earlier proved redundant is generated when proving redundancy of
$C'$. Then it makes sense to prove $C$ and $C'$ redundant
\ti{together}. In other words, instead of choosing between two extreme
paradigms of processing all clauses together (\apqe) or one clause at
a time (\Apqe), one can consider joint processing of \ti{subsets} of
target clauses. The second way to avoid looping is to completely avoid
adding new clauses with quantified variables. (So, no duplicates are
generated and looping becomes impossible.)  However, this requires a
radical overhaul of the current approach to PQE
solving~\cite{sac_by_pqe}.

%
% subsection
%
\subsection{On comparison of \Apqe with \apqe}
\label{ssec:triv}
Let $C$ be a clause to be taken out of the scope of quantifiers in
\prob{X}{F(X,Y)}. Let \pnt{y} be a full assignment to $Y$. Let $C$ be
redundant in \prob{X}{F} in subspace \pnt{y}. A trivial case of
redundancy occurs when $F \setminus \s{C}$ implies $C$ in subspace
\pnt{y}. (In general, this is not the case \ie redundancy of $C$ in
\prob{X}{F} in subspace \pnt{y} does not entail $F \setminus \s{C}
\imp C$ in this subspace.)  Even though we did not describe this fact,
\Apqe \ti{actually employs} a SAT-based procedure for identifying
trivial redundancies. How this is done is described
in~\cite{cav23}. On the other hand, \apqe \ti{does not use} any
procedure for identifying trivial redundancies. This partly explains
the poor performance of \apqe in comparison with \Apqe reported in
Subsection~\ref{ssec:comparison} and Appendix~\ref{ssec:exper1}.  Our
later results presented in~\cite{cav23} show that \apqe supplied with
a procedure for identifying trivial redundancies performs much better
on the examples we used in this report.

%
% subsection
%
\subsection{A few concluding thoughts}
\label{ssec:concl}
Despite the remarks above, we believe that processing one target
clause at a time is an appealing idea substantiated
experimentally. The remaining issue to be resolved is to find a way
for preventing \Apqe from generation of duplicate target clauses (thus
eliminating the possibility of looping).

%\clearpage
%\bibliographystyle{plain}
\bibliographystyle{IEEEtran}
\bibliography{short_sat,local,l1ocal_hvc}
\clearpage
%\appendices
\vspace{25pt}
\appendix
\noindent{\large \tb{Appendix}}
\section{Proofs}
 \setcounter{proposition}{0}
 \label{app:proofs}
Lemma~\ref{lem:blk_clause} is used in the proof of
Proposition~\ref{prop:blk_clause}.

\begin{lemma}
\label{lem:blk_clause}
Given a formula \prob{X}{F(X)}, let $C$ be a clause blocked in
\prob{X}{F} at $w$. Then $\prob{X}{F} \equiv \prob{X}{F \setminus
  \s{C}}$ i.e. $C$ is redundant in \prob{X}{F}.
\end{lemma}
\begin{proof}
Let us prove that $F$ and $F \setminus \s{C}$ are equisatisfiable (and
so $\prob{X}{F} \equiv \prob{X}{F \setminus \s{C}}$).  The
satisfiability of $F$ obviously implies that of $F \setminus
\s{C}$. Let us show that the converse is true as well.  Let \pnt{x} be
a full assignment to $X$ satisfying $F \setminus \s{C}$. If \pnt{x}
satisfies $C$, it satisfies $F$ and our proof is over.  Now assume
that \pnt{x} falsifies $C$ and hence falsifies $F$. Let \pent{x}{fl}
be the assignment obtained from \pnt{x} by flipping the value of
$w$.  (So \pent{x}{fl} satisfies $C$.) Let $G$ be the set of clauses
of $F$ resolvable with $C$ on $w$.  Let $w=b$ satisfy $C$ where $b \in
\s{0,1}$. (So, $w$ is assigned $b$ in \pent{x}{fl}, because \pnt{x}
falsifies $C$.)

First, let us show that \pent{x}{fl} satisfies $F \setminus G$. Assume
the contrary i.e. \pent{x}{fl} falsifies a clause $D$ of $F \setminus
G$. (Note that $D$ is different from $C$ because the latter is
satisfied by \pent{x}{fl}.) Assume that $D$ does not contain the
variable $w$. Then $D$ is falsified by the assignment \pnt{x} and
hence the latter does not satisfy $F \setminus \s{C}$. So we have a
contradiction.  Now, assume that $D$ contains $w$. Then $D$ is
resolvable with $C$ on $w$ and $D \in G$. So $D$ cannot be in $F
\setminus G$ and we have a contradiction again.

Since \pent{x}{fl} satisfies $F \setminus G$, then \Sub{(F \setminus
  G)}{w=b} is satisfiable.  By definition of a blocked clause (see
Definition~\ref{def:blk_cls}), $G$ is redundant in \prob{X}{F} in
subspace $w=b$.  So formula \Sub{F}{w=b} is satisfiable.  Hence $F$ is
satisfiable too.
\end{proof}

%
%  Proposition: blocked clause is es-implied
%

\begin{proposition}
  %\label{prop:blk_clause}
  Given a formula \prob{X}{F(X,Y)}, let $C$ be a clause blocked in
\prob{X}{F} at $w \in X$. Then $C$ is redundant in \prob{X}{F}
i.e. \prob{X}{F}\,$\equiv$ \prob{X}{F \setminus \s{C}}.  So,
$C$ is es-implied by $F \setminus \s{C}$ in \prob{X}{F}.
\end{proposition}
\begin{proof}
One needs to show that for every full assignment \pnt{y}\, to $Y$,
formulas \cof{F}{y} and \cof{(F \setminus \s{C})}{y} are
equisatisfiable.  If \pnt{y} satisfies $C$, it is trivially
true. Assume that \pnt{y} does not satisfy $C$. From
Definition~\ref{def:blk_cls} it follows that if $C$ is blocked in
\prob{X}{F} at a variable $w$, then \cof{C}{y} is blocked in
\cof{(\prob{F}{X})}{y} at $w$.  Then from Lemma~\ref{lem:blk_clause}
if follows that \cof{C}{y} is redundant in \cof{(\prob{F}{X})}{y}
\end{proof}

%
% Proposition: generating an es-implied clause from a blocked one
%
\begin{proposition}
  %\label{prop:blk_to_es}
  Given a formula \prob{X}{F(X,Y)}, let $\Sub{C}{trg} \in F$.  Let
\pnt{q} be an assignment to $X \cup Y$ that does not satisfy
\Sub{C}{trg}. Let \Sub{C}{trg} be blocked in \prob{X}{F} at $w \in X$
in subspace \pnt{q} where \mbox{$w \not\in \Va{q}$}. Let $l$($w$) be
the literal of $w$ present in \Sub{C}{trg}.  Let $K'$ denote the
longest clause falsified by \pnt{q}. Let $K''$ be a clause formed from
$l(w)$ and a subset of literals of \Sub{C}{trg} such that every clause
of \cof{F}{q} unresolvable with \cof{(\Sub{C}{trg})}{q} on $w$ is
unresolvable with \cof{(K'')}{q} too. Let $\Sub{K}{bct} = K' \vee
K''$. Then \mbox{$\cof{(\Sub{K}{bct})}{q} \imp
  \cof{(\Sub{C}{trg})}{q}$} and \Sub{K}{bct} is redundant in
\prob{X}{\Sub{K}{bct} \wedge (F \setminus \s{\Sub{C}{trg}})}.
\end{proposition}

\begin{proof}
The fact that \mbox{$\cof{(\Sub{K}{bct})}{q} \imp
  \cof{(\Sub{C}{trg})}{q}$} trivially follows from the definition of
\Sub{K}{bct}. The latter equals $K' \vee K''$ where $K'$ is falsified
by \pnt{q} and $K''$ consists only of (some) literals of
\Sub{C}{trg}. Now we prove that the clause \Sub{K}{bct} is redundant
in
\prob{X}{\Sub{K}{bct}\!\wedge\!(F\!\setminus\!\s{\Sub{C}{trg}})}. Let
$H$ denote $F \setminus \s{\Sub{C}{trg}}$.  One needs to show that for
every full assignment \pnt{y} to $Y$, \cof{(\Sub{K}{bct} \wedge H)}{y}
and \cof{H}{y} are equisatisfiable.  If \pnt{y} satisfies
\Sub{K}{bct}, it is trivially true. Let \pnt{y} not satisfy
\Sub{K}{bct}. (This means that the variables of $\Va{y} \cap \Va{q}$,
if any, are assigned the same value in \pnt{y} and \pnt{q}.) The
satisfiability of \cof{(\Sub{K}{bct} \wedge H)}{y} obviously implies
that of \cof{H}{y}. Below, we show \ti{in three steps} that the
converse is true as well. First, we introduce an assignment
\pent{p}{fl} such that $\pnt{y} \subseteq \pent{p}{fl}$ and $\pnt{q}
\subseteq \pent{p}{fl}$. (Here 'fl' stands for 'flipped'.)  Second, we
prove that \pent{p}{fl} satisfies $\cof{F}{q} \setminus \cof{G}{q}$
where $G$ is the set of clauses resolvable with \Sub{C}{trg}.  Third,
we show that the satisfiability of $\cof{F}{q} \setminus \cof{G}{q}$
and the fact that \Sub{C}{trg} is blocked imply that
\cof{(\Sub{K}{bct} \wedge H)}{y} is satisfiable.

\ti{Step} 1. Let \pnt{p}\, denote a full assignment to $X \cup Y$ such
that $\pnt{y} \subseteq \pnt{p}$\, and \pnt{p}\, satisfies \cof{H}{y}.
If \pnt{p}\, satisfies \cof{(\Sub{K}{bct} \wedge H)}{y}, our proof is
over.  Assume that \pnt{p}\, falsifies \cof{(\Sub{K}{bct} \wedge
  H)}{y}. Then \pnt{p}\, falsifies \Sub{K}{bct}. This means that
$\pnt{q} \subseteq \pnt{p}$.  Let $w=b$ denote the assignment to $w$
in \pnt{p}. Denote by \pent{p}{fl} the assignment obtained from
\pnt{p}\, by flipping the value of $w$ from $b$ to
$\overline{b}$. Denote by \pent{q}{ext} the extension $\pnt{y} \cup
\pnt{q} \cup \s{w=\overline{b}}$ of the assignment \pnt{q}.  Note that
$\pent{q}{ext} \subseteq \pent{p}{fl}$. Besides, due to the assignment
$w = \overline{b}$, both \pent{q}{ext} and \pent{p}{fl} satisfy
\Sub{K}{bct} and \Sub{C}{trg}.

\ti{Step} 2. Let $G$ denote the set of clauses of $F$ resolvable with
\Sub{C}{trg} on $w$.  Then \cof{G}{q} is the set of clauses of \cof{F}{q}
resolvable with \cof{(\Sub{C}{trg})}{q} on $w$. Let us show that
\pent{p}{fl} satisfies $\cof{F}{q} \setminus \cof{G}{q}$. Assume the
contrary i.e. there is a clause $D \in \cof{F}{q} \setminus
\cof{G}{q}$ falsified by \pent{p}{fl}. First, assume that $D$ does not
contain $w$. Then $D$ is falsified by \pnt{p} as well. So, \pnt{p}
falsifies \cof{F}{q} and hence \cof{H}{q} (because \cof{(\Sub{C}{trg})}{q}
 is satisfied by \pent{p}{fl} and so is different from $D$). Thus, we
 have a contradiction.  Now, assume that $D$ contains the literal
 $\overline{l(w)}$. Then it is resolvable with clause
 \cof{(\Sub{K}{bct})}{q}. This means that $D$ is resolvable with
 \cof{(\Sub{C}{trg})}{q}\, too. (By our assumption, every clause of
 \cof{F}{q} unresolvable with \cof{(\Sub{C}{trg})}{q} is unresolvable
 with \cof{(\Sub{K}{bct})}{q}\, too.) Then $D$ cannot be in
 $\cof{F}{q} \setminus \cof{G}{q}$ and we have a contradiction.

\ti{Step} 3. \vspace{2pt} Since \pent{p}{fl} satisfies $\cof{F}{q}
\setminus \cof{G}{q}$, the formula $\ccof{F}{q}{ext} \setminus
\ccof{G}{q}{ext}$ is satisfiable.  The same applies to \ccof{(F
  \setminus G)}{q}{ext}. Since \Sub{C}{trg} is blocked in \prob{X}{F} at
$w$ in subspace \pnt{q}, it is also blocked in \prob{X}{F} in subspace
$\pnt{y} \cup \pnt{q}$.  Then \ccof{(F \setminus G)}{q}{ext} es-implies
\ccof{G}{q}{ext} (see Definition~\ref{def:blk_cls}) and \ccof{F}{q}{ext} is
satisfiable too. Since \Sub{K}{bct} is satisfied by \pent{q}{ext},
then \ccof{(\Sub{K}{bct} \wedge F)}{q}{ext} is satisfiable. Hence
\cof{(\Sub{K}{bct} \wedge F)}{y} is satisfiable and so is
\cof{(\Sub{K}{bct} \wedge H)}{y}
\end{proof}

%\input{u0seful_lemma}

%\vspace{-1pt}
\section{Backtracking By \Apqe}
%\section{Certificates In More Detail}
%\section{More Details About Using Certificates}
%\label{sec:ncnfl_certs}
\label{app:bcktr}

When a SAT-solver encounters a conflict, it generates a conflict
clause and backtracks to the smallest decision level where this clause
is unit. So, an assignment can be derived from this clause.  In
contrast to a SAT-solver, the goal of a PQE-solver is to prove a
target clause \Sub{C}{trg} redundant rather than find a satisfying
assignment. So, \Apqe backtracks slightly differently from a
SAT-solver. After \Apqe derives a certificate $K$ proving \Sub{C}{trg}
in the current subspace, it backtracks to the smallest decision level
at which the \ti{conditional} of the derived certificate $K$ (rather
than $K$ itself) is unit.

%
% definition
%
%\vspace{-10pt}

\begin{definition}
  \label{def:cond}
Let $K$ be a certificate stating the redundancy of clause \Sub{C}{trg}
in a subspace. The clause consisting of the literals of $K$ not shared
with \Sub{C}{trg} is called the \tb{conditional of} \bm{K}.
\end{definition}

%
% second example of start's operation
%
\setlength{\intextsep}{4pt}
\setlength{\textfloatsep}{10pt}
\begin{wrapfigure}{l}{1.5in}
%\begin{figure}
%\begin{center}
\small
%\scriptsize
%\vspace{5pt}
\begin{tabbing}
aaaa\=bb\=cc\=dd\= \kill

\scriptsize{1}\>  Find $F^*_1(Y)$ such that \\
\scriptsize{2}\>   \pqe{X}{F^*_1}{F_1}{F_2}  \\
\scriptsize{3}\> $Y\!=\!\s{y_1}$,\,$X\!=\!\s{x_2,x_3,x_4,\dots}$\\
\scriptsize{4}\> $F_1 = \s{C_1}$, $C_1=x_2 \vee x_4$\\
\scriptsize{5}\> $F_2\!=\!\s{C_2,\!C_3,\!\dots}$,$C_2\!=\!y_1\!\vee\!x_3$,\\
\scriptsize{6}\Tt $C_3\!=\!\overline{x}_3\!\vee\!x_4,\dots$ \\[4pt]

\scriptsize{7}\>  pick. $C_1 \in F_1$ to prove red. \\[4pt]
\scriptsize{8}\>    --- call \prd ---\\
\scriptsize{9}\> decision: $y_1=0$  at level 1 \\
\scriptsize{10}\> \ti{BCP}:$(C_2\!:x_3\!=\!1)$ \\
\scriptsize{11}\ttt $C_3 = x_4$ in curr. subsp.\\[4pt]
\scriptsize{12}\> LEAF: $C_3$ impl. $C_1$ at  level 1\\
\scriptsize{13}\> $K_1 = y_1 \vee x_4$ is the final cert. \\
\scriptsize{14}\> $K_1 = \mi{Res}(C_3,C_2,x_3)$ \\ [4pt]

\scriptsize{15}\> backtracking to level 0 \\
\scriptsize{16}\> \ti{BCP}:  $(K_1:\!y_1=1)$ \\[4pt]
\scriptsize{17}\>  .... \\

\end{tabbing} 
\vspace{-15pt}
%\caption{Operation of \Apqe}
\caption{Backtracking by \Apqe}
%\vspace{-15pt}
\label{fig:exmp2}
%\end{figure}
\end{wrapfigure}

%\vspace{-5pt}
If the conditional of $K$ is empty, $K$ implies \Sub{C}{trg} in the
entire space. Otherwise, $K$ implies \Sub{C}{trg} only in subspaces
\pnt{q} where the conditional of $K$ is falsified by \pnt{q}.  One can
derive an implied assignment from $K$ when its \tb{conditional is
  unit} like this is done by a SAT-solver when a clause becomes unit.

%
%  Example
%
\begin{example}
Consider the example of Fig.~\ref{fig:exmp2} showing the operation of
\Apqe. This figure gives only the relevant part of formula $F_2$ and
the relevant fragment of the execution trace.  \prd begins proving the
target clause $C_1=x_2 \vee x_4$ redundant by the decision assignment
$y_1 = 0$.  Then it calls \ti{BCP} that derives $x_3=1$ from the
clause $C_2$.  At this point, $C_3$ becomes the unit clause $x_4$
implying $C_1$. So, \ti{BCP} returns $C_3$ as the reason for
backtracking (i.e. as the clause \Sub{K}{bct} in line 6 of
Fig.~\ref{fig:prv_red}). Then the \lrn procedure generates the final
certificate \mbox{$K_1 = y_1 \vee x_4$} by resolving $C_3$ and $C_2$
to drop the non-decision variable $x_3$ assigned at level 1 (line 14).

The \ti{conditional} of $K_1$ is the unit clause $y_1$ because the
literal $x_4$ is shared by $K_1$ and the target clause $C_1$.  \prd
backtracks to level 0, the smallest level where the conditional of
$K_1$ is unit. (Note that $K_1$ itself is not unit at level 0). Then
\prd runs \ti{BCP} and derives the assignment $y_1\!=\!1$ from $K_1$
even though $K_1$ is not unit at level 0.  This derivation is possible
because $K_1$ certifies that the redundancy of $C_1$ in subspace
$y_1\!=\!0$ is \ti{already proved} $\blacksquare$
\end{example}
\vspace{-5pt}

As we mentioned above, in the general case, after deriving a
certificate $K$, \prd backtracks to the smallest decision level where
the conditional of $K$ is unit. The assignment derived from $K$ is
added to the assignment queue $Q$ (lines 10 and 12 of
Fig.~\ref{fig:prv_red}). If $K$ shares no literals with \Sub{C}{trg},
the \prd procedure backtracks as a regular SAT-solver, i.e. to the
smallest decision level where $K$ is unit.

%\subsection{Certificates added to the formula}
%\label{ssec:str_cls}
%\vspace{-5pt}

\clearpage
\section{Operation Of \Apqe When \Sub{C}{trg} Becomes Unit}
\label{app:exmp}

%
% third example of start's operation
%
\setlength{\intextsep}{4pt}
\setlength{\textfloatsep}{10pt}
\begin{wrapfigure}{l}{2.1in}
%\begin{figure}
%\begin{center}
\small
%\scriptsize
%\vspace{5pt}
\begin{tabbing}
aaa\=bb\=cc\=dd\= \kill

\scriptsize{1}\>  Find $F^*_1(Y)$ such that \\
\scriptsize{2}\>   \pqe{X}{F^*_1}{F_1}{F_2}  \\
\scriptsize{3}\> $Y\!=\!\s{y_1}$,\,$X\!=\!\s{x_2,x_3}$\\
\scriptsize{4}\> $F_1 = \s{C_1}$, $C_1=y_1 \vee x_2$\\
\scriptsize{5}\> $F_2\!=\!\s{C_2,C_3}$,$C_2\!=\!\overline{x}_2\!\vee\!x_3$,\\
\scriptsize{6}\> $C_3\!=\!\overline{y}_1\!\vee\!\overline{x}_3$ \\[4pt]

\scriptsize{7}\> putting $C_1$ to target level $A$\\
\scriptsize{8}\>  \bm{\mi{PrvRed}}: proving $C_1$ redund. \\
\scriptsize{9}\> dec.: $y_1=0$  at dec. level 1 \\
\scriptsize{10}\Tt \ti{BCP}: $C_1$ is unit at dec. level 1\\
\scriptsize{11}\Tt creating target level $B$ \\
\scriptsize{12}\Tt of clauses res. with $C_1$ on $x_2$\\
\scriptsize{13}\Tt making impl. assign. $x_2=1$ \\
\scriptsize{14}\Tt picking $C_2$ as a new targ.\\[4pt]
\scriptsize{15}\Tt \bm{\mi{PrvRed}}: prov. $C_2$ redund. \\ 
\scriptsize{16}\Tt in subsp. $y_1=0,x_2=1$ \\
\scriptsize{17}\Tt LEAF: $C_2$ is blocked at $x_3$\\
\scriptsize{18}\Tt Der. cert. $K'\!=\!y_1\!\vee\!x_3$  \\
\scriptsize{19}\Tt $C_2$ is red. in subsp. above \\
\scriptsize{20}\Tt $C_2$ is temporarily removed \\
\scriptsize{21}\> eliminating targ. level B\\
\scriptsize{22}\> undoing $x_2 = 1$\\
\scriptsize{23}\>  $C_2$ is restored in the formula\\
\scriptsize{24}\> LEAF: $C_1$ is blocked at $x_2$ \\
\scriptsize{25}\> der. cert. $K''\!= y_1\!\vee\!x_2$\\
\ttt                      $\cdot\cdot\cdot$ \\
%% \tb{\scriptsize{17}}\> \\
%% \tb{\scriptsize{17}}\> \\
%% \tb{\scriptsize{17}}\> \\
%% \tb{\scriptsize{17}}\> \\
\end{tabbing} 
\vspace{-23pt}
%\caption{Operation of \Apqe}
%\caption{Example of secondary targets}
\caption{\Sub{C}{trg} becomes unit}
%\vspace{-40pt}
\label{fig:exmp3}
%\end{figure}
\end{wrapfigure}

In Section~\ref{sec:trg_unit}, we described how \Apqe operates when
the current target clause \Sub{C}{trg} becomes unit. In this appendix,
we give a concrete example.  Consider solving the PQE problem shown in
Fig.~\ref{fig:exmp3} by lines 1-6. First, $C_1$ is picked as a clause
to prove. We will refer to it as the \tb{primary} target assuming that
it makes up target level $A$. After decision assignment $y_1=0$, the
clause $C_1$ turns into unit clause $x_2$ (lines 9-10). \tb{Denote}
the current assignment (i.e. $y_1=0$) as \pnt{q}. At this point, a
SAT-solver would simply derive the assignment $x_2=1$. However, the
goal of \prd is not to check if $F_1 \wedge F_2$ is satisfiable but to
prove $C_1$ redundant. The fact that $C_1$ is falsified in subspace
$\pnt{q}\,\cup \s{x_2=0}$ does not say anything about whether $C_1$ is
\ti{redundant} there.

So, \prd creates a new target level (referred to as level $B$).  It
consists of the clauses resolvable with $C_1$ on $x_2$. Suppose all
clauses of this level are redundant in subspace $\pnt{q}\, \cup
\s{x_2=1}$. Then according to Definition~\ref{def:blk_cls}, $C_1$ is
blocked (and hence redundant) in \prob{X}{F_1 \wedge F_2} in subspace
\pnt{q}. In our case, level $B$ consists only of $C_2$. So, \prd
recursively calls itself to prove redundancy of $C_2$ in subspace
$\pnt{q}\,\cup \s{x_2=1}$ (lines 15-20). Note that $C_2$ is blocked at
$x_3$ in this subspace since $C_3$ (the clause resolvable with $C_2$
on $x_3$) is satisfied by $y_1=0$. Then using
Proposition~\ref{prop:blk_to_es}, \prd derives the certificate $K' =
y_1 \vee x_3$ asserting the redundancy of $C_2$. The latter is
temporarily removed from the formula as redundant (see
Remark~\ref{rem:loc_rem}). At this point, the second activation of \prd
terminates.

Then the first activation of \prd undoes target level $B$ and
assignment $x_2\!=\!1$. The clause $C_2$ is restored in the formula
(lines 21-23).  Now, the primary target $C_1$ is blocked at $x_2$,
since $C_2$ is proved redundant in subspace $\pnt{q}\,\cup
\s{x_2\!=\!1}$. Using Proposition~\ref{prop:blk_to_es}, \prd derives
the certificate \mbox{$K''\!=\!y_1\!\vee\!x_2$} proving redundancy of
$C_1$ in the entire space.

%\section{A Special Case}
%\section{Conflicts Involving Non-Conflict Certificates}
\section{Certificate Generation When A Conflict Occurs}
\label{app:spec_case}

In this appendix, we discuss in more detail the generation of a
certificate by the \lrn procedure when a conflict occurs. As before,
we denote $F_1 \wedge F_2$ by $F$.  Let \Sub{C}{trg} be the current
target clause. Let \Sub{K}{bct} be the clause of $F$ falsified in this
conflict. (Here, we use the notation of Figure~\ref{fig:prv_red}
describing the \prd procedure).  First, consider the case when
\mbox{\Sub{K}{bct} $\neq$ \Sub{C}{trg}}. Then \lrn generates a
certificate $K$ as described in Subsection~\ref{ssec:prv_red}. Namely,
it starts with \Sub{K}{bct} gradually resolving out literals assigned
at the conflict level by non-decision assignments. Since \Sub{C}{trg}
is not involved in derivation of $K$, the latter is a witness
certificate.

%
% fifth example of start's operation
%
\setlength{\intextsep}{4pt}
\setlength{\textfloatsep}{10pt}
\begin{wrapfigure}{l}{1.6in}
%\begin{figure}
%\begin{center}
\small
%\scriptsize
%\vspace{5pt}
\begin{tabbing}
aaa\=bb\=cc\=dd\= \kill

\scriptsize{1}\>  Find $F^*_1(Y)$ such that \\
\scriptsize{2}\>   \pqe{X}{F^*_1}{F_1}{F_2}  \\
\scriptsize{3}\> $Y\!=\!\s{y_1}$,\,$X\!=\!\s{x_2,x_3}$\\
\scriptsize{4}\> $F_1 = \s{C_1}$, $C_1=\overline{x}_2 \vee x_3$\\
\scriptsize{5}\> $F_2\!=\!\s{C_2,\!C_3}$,$C_2\!=\!\overline{y}_1\!\vee\!x_2$,\\
\scriptsize{6}\Tt $C_3\!=\!\overline{y}_1\!\vee\!\overline{x}_3$ \\[4pt]

\scriptsize{7}\>  pick. $C_1 \in F_1$ to prove red. \\[4pt]
\scriptsize{8}\>    --- call \prd ---\\
\scriptsize{9}\> decision: $y_1=0$  at level 1 \\[4pt]
\scriptsize{10}\> LEAF: $C_1$ is blocked at $x_2$\\
\scriptsize{11}\> (since $C_2$ is sat. by $y_1$ = 0)\\
\scriptsize{12}\> $K_1\!=\!y_1\!\vee\!\overline{x}_2$ is\,a\,witness\,cert.\\
\scriptsize{13}\> $K_1$ is not added to $F_1\!\wedge\!F_2$\\[4pt]

\scriptsize{14}\> backtracking to level 0 \\
\scriptsize{15}\> \ti{BCP}:  $(K_1:\!y_1=1)$ \\
\scriptsize{16}\ttt $(C_2\!:x_2\!=\!1)(C_3\!:x_3\!=\!0)$ \\[4pt]
\scriptsize{17}\> LEAF: conflict at  level 0\\
\scriptsize{18}\> $C_1\!=\!\overline{x}_2\!\vee\!x_3$ is falsified\\
\scriptsize{19}\> $\hat{K} = \overline{y}_1$ is a new clause\\
\scriptsize{20}\Tt falsif. in curr. subspace \\
\scriptsize{21}\Tt $R_1=\mi{Res}(C_1,C_2,x_2)$,\\
\scriptsize{22}\Tt $\hat{K}=\mi{Res}(R_1,C_3,x_3)$ \\
\scriptsize{23}\> $\hat{K}$ is added to $F_1\!\wedge\!F_2$\\
\scriptsize{24}\> $K_2\!=\!\overline{x}_2$ is a witness cert.\\
\scriptsize{25}\Tt $K_2\!=\!\mi{Res}(K_1,\!\hat{K},\!y_1)$ \\
\scriptsize{26}\>  $K_2$ is not added to $F_1\!\wedge\!F_2$\\
\ttt                      $\cdot\cdot\cdot$ \\
\end{tabbing} 
\vspace{-20pt}
%\caption{\Apqe, mixed conflict, the current target is involved}
\caption{Adding a special clause after a conflict}
%\vspace{-10pt}
\label{fig:exmp5}
%\end{figure}
\end{wrapfigure}

Now, consider the case when \Sub{K}{bct} = \Sub{C}{trg}. If, no
relevant assignment is derived from a witness certificate, \lrn
generates the resulting certificate $K$ as described above. Since
\Sub{C}{trg} is involved in derivation of $K$, the latter is a
\ti{participant} certificate that is added to the formula. If an
assignment relevant to the conflict is derived from a witness
certificate, \lrn acts differently. Namely, it derives a \ti{witness}
certificate $K$ \ti{and} a special clause $\hat{K}$ that is added to
the formula.  (For the sake of simplicity, we did not mention this
fact in the pseudo-code of the \prd procedure.)

Figure~\ref{fig:exmp5} illustrates adding a special clause. Here
$C_1\!=\!  \overline{x}_2\!\vee\!x_3$ is the target clause. In the
branch $y_1=0$, \prd proves $C_1$ redundant by deriving a witness
certificate $K_1\!=\!y_1 \vee\, \overline{x}_2$ (lines 9-13).  Then
\prd backtracks to level 0 and runs \ti{BCP} to derive $y_1\!=\!1$
from $K_1$, $x_2\!=\!1$ from $C_2$ and $x_3\!=\!0$ from $C_3$. At this
point, $C_1$ is falsified i.e. a conflict occurs. Assume we construct
a certificate $K_2\!=\!\overline{x}_2$ by resolving $C_1$ with $C_2$,
$C_3$, and $K_1$ (i.e. with the clauses from which the relevant
assignments were derived). Then we have a problem. On one hand, $K_2$
is a participant certificate that has to be added to $F$ since the
target clause $C_1$ was involved in building $K_2$. On the other hand,
$K_2$ may not be implied by $F$ since a witness certificate $K_1$ was
involved in producing $K_2$.  (A witness certificate $K$ is, in
general, only es-implied by $F$ in \prob{X}{K \wedge F}.) This breaks
the invariant maintained by \Apqe that only clauses implied by $F$ are
added to it.

The \lrn procedure addresses the problem above as follows. First, it
generates a clause $\hat{K}=\overline{y}_1$ that is falsified in the
current subspace and so ``replaces'' $C_1$ as the reason for the
conflict.  $\hat{K}$ is built without using witness certificates and
so \ti{can} be added to $F$. It is obtained by resolving $C_1$ with
$C_2$ and $C_3$ and is added to $F$ (lines 19-23). Then \lrn derives
the certificate $K_2 = \overline{x}_2$ by resolving $\hat{K}$ and
$K_1$. The clause $K_2$ certifies the redundancy of the target clause
$C_1$ in the entire space. Note that $K_2$ was derived using $\hat{K}$
instead of the target clause $C_1$. So, it is a witness certificate
that does not have to be added to the formula.

Here is how one handles the general case when \Sub{K}{bct} =
\Sub{C}{trg} and a witness certificate is involved in the
conflict. First, one produces a special clause $\hat{K}$.  It is
obtained by resolving \Sub{C}{trg} with clauses of $F$ from which
relevant assignments were derived. This process \ti{stops} when the
assignment derived from a \ti{witness certificate} is reached. Then
$\hat{K}$ is added to the formula $F$. (This can be done since witness
certificates are not used in derivation of $\hat{K}$.) After that,
\lrn generates a certificate $K$ starting with $\hat{K}$ as a clause
falsified in the current subspace. (That is, $\hat{K}$ replaces
\Sub{C}{trg} as the cause of the conflict.) Since \Sub{C}{trg} is not
involved in generation of $K$, the latter is a witness certificate.

\section{Correctness of \Apqe}
\label{app:sound_compl}
In this appendix, we give a proof that \Apqe is correct.  Let \Apqe be
used to take $F_1$ out of the scope of quantifiers in
$\exists{X}[F_1(X,Y) \wedge F_2(X,Y)]$. We will denote $F_1 \wedge
F_2$ by $F$. In Subsection~\ref{ssec:sound}, we show that \Apqe is
sound. Subsection~\ref{ssec:mstart} discusses the problem of
generating duplicate clauses by \Apqe and describes a solution to this
problem.  In Subsection~\ref{ssec:compl}, we show that the versions of
\Apqe that do not produce duplicate clauses are complete.

%\input{o2ld_soundness}
%\section{Correctness Of \Apqe}
\subsection{\Apqe is sound}
\label{ssec:sound}

In its operation, \Apqe adds participant
certificates and removes target clauses from $F$.  Denote the initial
formula $F$ as \Sup{F}{ini}. Let \pnt{y} be a full assignment to the
variables of $Y$ (i.e. unquantified ones). Below, we demonstrate that
for every subspace \pnt{y}, \Apqe preserves the equisatisfiability
between the \Sup{F}{ini} and the current formula $F$. That is,
\mbox{\prob{X}{\Sup{F}{ini}} $\equiv$ \prob{X}{F}}. Then we use this
fact to show that \Apqe produces a correct solution.

First, consider adding participant certificates by \Apqe. As we
mention in Section~\ref{sec:start}, every clause added to $F$ is
\ti{implied} by $F$. (If a clause $K$ is es-implied by $F$ in
\prob{X}{K \wedge F}, it is used only as a \ti{witness} certificate
and is not added to $F$.) So, adding clauses cannot break
equisatisfiability of $F$ and \Sup{F}{ini} in a subspace \pnt{y}.

Now, we consider removing target clauses from $F$ by \Apqe.  A target
clause \Sub{C}{trg} is permanently removed from the formula only if a
certificate $K$ implying \Sub{C}{trg} in the entire space is
derived. $K$ is obtained by resolving clauses of the current formula
$F$ and witness certificates (if any). Derivation of $K$ is correct
due to correctness of Proposition~\ref{prop:blk_to_es} (describing
generation of clauses that are es-implied rather than implied by the
formula) and soundness of resolution.  So, removing \Sub{C}{trg} from
$F$ cannot break equisatisfiability of $F$ and \Sup{F}{ini} in some
subspace \pnt{y}.

The fact that $\prob{X}{\Sup{F}{ini}} \equiv \prob{X}{F}$ entails that
\Apqe produces a correct solution. Indeed, \Apqe terminates when the
current formula $F_1$ does not contain a quantified clause.  So, the
\ti{final} formula $F$ can be represented as $F_1(Y) \wedge
F_2(X,Y)$. Then \mbox{\prob{X}{F^{\mi{ini}}_1 \wedge F^{\mi{ini}}_2}
  $\equiv F_1 \wedge$ \prob{X}{F_2}}. \Apqe does not add any clauses
to $F_2$. Hence, the final and initial formulas $F_2$ are
identical. So, $\prob{X}{F^{\mi{ini}}_1 \wedge F^{\mi{ini}}_2} \equiv
F^*_1 \wedge \prob{X}{F^{\mi{ini}}_2}$ where $F^*_1$ is the final
formula $F_1$.

%
% Subsection
%
\subsection{Avoiding generation of duplicate clauses}
\label{ssec:mstart}
The version of \prd described in
Sections~\ref{sec:exmp1}-\ref{sec:trg_unit} may generate a duplicate
of a quantified clause that is currently \ti{proved redundant}. To
avoid generating duplicates one can modify \Apqe as follows.  (We did
not implement this modification due to its inefficiency. We present it
just to show that the problem of duplicates can be fixed in
principle.)  We will refer to this modification as \Mapqe.

Suppose \prd generated a quantified clause $C$ proved redundant
earlier. This can happen only when all variables of $Y$ are assigned
because they are assigned before those of $X$. Then \Mapqe discards
the clause $C$, undoes the assignment to $X$, and eliminates all
recursive calls of \prd. That is \Mapqe returns to the original call
of \prd made in the main loop (Fig.~\ref{fig:mloop}, line 7).  Let
\Sub{C}{trg} be the target clause of this call of \prd and \pnt{y}\,
be the current (full) assignment to $Y$. At this point \Mapqe calls an
internal SAT-solver to prove redundancy of \Sub{C}{trg} in subspace
\pnt{y}.  This goal is achieved by this SAT-solver via generating a
witness or participant certificate implying \Sub{C}{trg} in subspace
\pnt{y} (see below). After that, \prd goes on as if it just finished
line 10 of Figure~\ref{fig:prv_red}.

Let $B(Y)$ denote the longest clause falsified by \pnt{y}. Suppose the
internal SAT-solver of \Mapqe proves \cof{F}{y} unsatisfiable. (Recall
that $F$ denotes $F_1 \wedge F_2$.)  Then the clause $B$ is a
certificate of redundancy of \Sub{C}{trg} in \cof{F}{y}. If
\Sub{C}{trg} is involved in proving \cof{F}{y} unsatisfiable, $B$ is a
participant certificate. The \prd procedure adds $B$ to $F$ to make
\Sub{C}{trg} redundant in subspace \pnt{y}.  If \cof{F}{y} is proved
unsatisfiable without using \Sub{C}{trg}, then $B$ is a witness
certificate that is not added to $F$.

Suppose that \cof{F}{y} is satisfiable. Then the internal SAT-solver
above derives an assignment \pnt{p} satisfying \cof{F}{y} where
$\pnt{y} \subseteq \pnt{p}$. Note that \pnt{y} does not satisfy
\Sub{C}{trg} since, otherwise, \prd would have already proved
redundancy of \Sub{C}{trg} in subspace \pnt{y}. Hence, \pnt{p}
satisfies \Sub{C}{trg} by an assignment to a variable $w \in X$. Then
\prd derives a witness certificate $K$ equal to $B \vee l(w)$ where
$l(w)$ is the literal of $w$ present in \Sub{C}{trg}. It is not hard
to show that $K$ is indeed a certificate. First, it implies
\Sub{C}{trg} in subspace \pnt{y} certifying its redundancy
there. Second, $K$ is es-implied by $F \setminus \s{\Sub{C}{trg}}$ in
\prob{X}{K \wedge (F \setminus \s{\Sub{C}{trg}})}.

%
%  Subsection: completeness
%
\vspace{-8pt}
\subsection{\Apqe is complete}
\label{ssec:compl}
In this subsection, we show the completeness of the versions of \Apqe
that do not generate duplicate clauses. (An example of such a version
is given in the previous subsection).  The completeness of \Apqe
follows from the fact that
\begin{itemize}
\item the number of times \Apqe calls the \prd procedure (to prove redundancy of the
  current target clause) is finite;
\item the number of steps performed by one call of \prd is finite.
\end{itemize}

So, \Apqe always terminates.  First, let us show that \prd is called a
finite number of times. By our assumption, \Apqe does not generate
quantified clauses seen before. So, the number of times \prd is called
in the main loop of \Apqe (see Figure~\ref{fig:mloop}) is finite. \prd
recursively calls itself when the current target clause \Sub{C}{trg}
becomes unit. The number of such calls is finite (since the number of
clauses that can be resolved with \Sub{C}{trg} on its unassigned
variable is finite). The depth of recursion is finite here. Indeed,
before a new recursive call is made, the unassigned variable $w \in X$
of \Sub{C}{trg} is assigned and $X$ is a finite set.  Summarizing, the
number of recursive calls made by \prd invoked in the main loop of
\Apqe is finite.

Now we prove that the number of steps performed by a single call of
\prd is finite.  (Here we ignore the steps taken by \ti{recursive}
calls of \prd.) Namely, we show that \prd examines a finite search
tree.  The number of branching nodes of the search tree built by \prd
is finite because $X \cup Y$ is a finite set. Let us show that \prd
indeed builds a tree. That is \prd does not have ``\ti{holes}'' and
always reaches a \ti{leaf} i.e. a node where a backtracking condition
is met. Below, we list the four kinds of leafs reached by \prd.  (The
backtracking conditions are identified by the \ti{BCP} procedure
called by \prd.)  Let \pnt{q} specify the current assignment to $Y
\cup X$. A leaf of the first kind is reached when the target clause
\Sub{C}{trg} becomes unit in subspace \pnt{q}. Then \ti{BCP} calls the
\ti{Rcrs} procedure (line 11 of Fig.~\ref{fig:bcp}) and \prd
backtracks. \prd reaches a leaf of the second kind when \ti{BCP} finds
a clause of $F$ implying \Sub{C}{trg} in subspace \pnt{q}. A leaf of
the third kind is reached when \ti{BCP} identifies a clause falsified
by \pnt{q} (i.e. a conflict occurs). \prd reaches a leaf of the fourth
kind when the target clause \Sub{C}{trg} is blocked in \prob{X}{F} in
subspace \pnt{q}.

If $F$ is unsatisfiable in subspace \pnt{q}, \prd always reaches a
leaf before all variables of $Y \cup X$ are assigned. (Assigning all
variables without a conflict, i.e. without reaching a leaf of the
third kind, would mean that $F$ is satisfiable in subspace \pnt{q}.)
Let us show that if $F$ is \ti{satisfiable} in subspace \pnt{q}, \prd
also always reaches a leaf before every variable of $Y \cup X$ is
assigned. (That is before a satisfying assignment is generated.)  Let
\pnt{p}\, be an assignment satisfying $F$ where $\pnt{q} \subseteq
\pnt{p}$. Consider the worst case scenario. That is all variables of
$Y \cup X$ but some variable $w$ are already assigned in \pnt{q} and
no leaf condition is encountered yet. Assume that no literal of $w$ is
present in the target clause \Sub{C}{trg}. Since \pnt{q} contains all
assignments of \pnt{p} but that of $w$, \Sub{C}{trg} is satisfied by
\pnt{q}.  Recall that \prd does not make \ti{decision} assignments
satisfying \Sub{C}{trg} (see Subsection~\ref{ssec:prv_red}). So,
\Sub{C}{trg} is satisfied by an assignment \ti{derived} from a clause
$C$. Then $C$ implies \Sub{C}{trg} in subspace \pnt{q} and a leaf of
the second kind must have been reached. So, we have a contradiction.

Now assume that \Sub{C}{trg} has a literal $l(w)$ of $w$.  Note that
since \prd assigns variables of $Y$ before those of $X$, then $w \in
X$. Since \Sub{C}{trg} is not implied by a clause of $F$ in subspace
\pnt{q}, all the literals of \Sub{C}{trg} but $l(w)$ are falsified by
\pnt{q}.  Let us show that \Sub{C}{trg} is blocked in \prob{X}{F} at
$w$ in subspace \pnt{q}.  Assume the contrary i.e. there is a clause
$C$ resolvable with \Sub{C}{trg} on $w$ that contains the literal
$\overline{l(w)}$ and is not satisfied yet. That is all the literals
of $C$ other than $\overline{l(w)}$ are falsified by \pnt{q}. Then
\pnt{p} cannot be a satisfying assignment because it falsifies either
\Sub{C}{trg} or $C$ (depending on how the variable $w$ is
assigned). So, we have a contradiction.Thus, \Sub{C}{trg} is blocked
at $w$ in subspace \pnt{q} and hence a leaf of the fourth kind is
reached.

%\clearpage
%\vspace{-10pt}
\section{Experiments With HWMCC-13 Benchmarks}
\label{app:exper2}
In this appendix, we describe experiments with multi-property
benchmarks of the HWMCC-13 set~\cite{hwmcc13}. (We use this set
because the multi-property track has been discontinued in HWMCC since
2013.) Each benchmark consists of a sequential circuit $N$ and
invariants that are supposed to hold for $N$. One can view the
conjunction of those invariants as a \tb{specification} \spe for $N$.
In the experiments, we used \Apqe to generate invariants of $N$ not
implied by \spe.  Similarly to the experiment of
Section~\ref{sec:exper1}, the formula $F_k = I(S_0) \wedge T(S_0,S_1)
\wedge \dots \wedge T(S_{k-1},S_k)$ was used to generate
invariants. The number $k$ of time frames was in the range of
\mbox{$2\!\leq\!k\!  \leq\!10$}. Specifically, we set $k$ to the
largest value in this range where $|F_k|$ did not exceed 500,000
clauses. We discarded the benchmarks with $|F_2|\!>\!500,000$. So, in
the experiments, we used 112 out of the 178 benchmarks of the set.

We describe three experiments. In every experiment, we generated local
invariants $H_k$ by taking out a clause of \prob{\Abs{k-1}}{F_k}.  The
objective of the \ti{first experiment} was to demonstrate that \Apqe
could compute $H_k$ for realistic designs.  We also showed in this
experiment that PQE could be much easier than QE and that \Apqe
outperforms our previous PQE-solver called \apqe.  The \ti{second
  experiment} demonstrated that a clause $Q$ of a local invariant
$H_k$ generated by \Apqe was often a global invariant not implied by
the specification \spe.  (As we mentioned in
Section~\ref{sec:inv_gen}, the necessary condition for an invariant
$Q$ to be unwanted is $\spe \not\imp Q$.) Note that the circuits of
the HWMCC-13 set are ``\tb{anonymous}''. So, we could not decide if
$Q$ was an unwanted invariant. Our goal was to show that \Apqe was
good enough to generate invariants not implied by \spe. (Then one
could check those invariants for being unwanted as described in
Section~\ref{sec:be_test}.) As in the experiment of
Section~\ref{sec:exper1}, we took out only clauses containing a state
variable of the $k$-th time frame.  The choice of the next clause to
take out was made according to the order in which clauses were listed
in $F_k$. In the \ti{third experiment}, we showed that \Apqe generates
invariants that are different from those produced by \ict.

%
% Subsection: experiments in more detail
%  Experiment 1
%
%\vspace{-10pt}
\clearpage
\subsection{Experiment 1}
\vspace{-10pt}
\label{ssec:exper1}
%
%
% ST-PQE versus DS-PQE
%
\vspace{5pt}
\begin{wraptable}{l}{1.8in}
%\begin{table}
\centering
%\small
  % \vspace{15pt}
\scriptsize
\captionsetup{justification=centering}
\caption{\small{\Apqe and \apqe. The time limit is 5\,sec.}}
\vspace{-5pt}
%\begin{center}
%\renewcommand{\arraystretch}{1.2} % Default value: 1
%\begni{tabular}{|p{22pt}|p{36pt}|c|c|c|c|} \hline
  \begin{tabular}{|p{25pt}|p{25pt}|p{25pt}|p{25pt}|} \hline pqe &
total & sol- & unsol- \\ solver & probl. & ved & ved \\ \hline
$\mi{start}$ & 5,418 & \tb{3,102} & \tb{2,316} \\ \hline
$\mi{ds}$-$\mi{pqe}$ & 5,418 & 1,285 & 4,133 \\ \hline
\end{tabular}                
%\end{center}
\vspace{10pt}
\label{tbl:old_new}
%\end{table}
\end{wraptable}

%
%\tb{Experiment 1.}
In this experiment, for each benchmark out of 112 mentioned above we
generated PQE problems of taking a clause out of
\prob{\Abs{k-1}}{F_k}. Some of them were trivially solved by
pre-processing. The latter eliminated the blocked clauses of $F_k$
that could be easily identified and ran BCP launched due to the unit
clauses specifying the initial state. We generated up to 50
\ti{non-trivial} problems per benchmark ignoring those solved by
pre-processing. (For some benchmarks the total number of non-trivial
problems was under 50.)

We compared \Apqe with \apqe introduced in~\cite{hvc-14} that is based
on the machinery of D-sequents. The relation of D-sequents and
certificates is briefly discussed in Section~\ref{sec:bg}. In contrast
to \Apqe, \apqe proves redundancy of many targets \ti{at once}, which
can lead to generating very deep search trees.  To make the experiment
less time consuming, we limited the run time of \Apqe to 5 sec. per
PQE problem. The results are shown in Table~\ref{tbl:old_new}. The
first column gives the name of a PQE solver. The second column shows
the total number of PQE problems we generated for the 112
benchmarks. The last two columns give the number of problems solved
and unsolved in the time limit. Table~\ref{tbl:old_new} shows that
\Apqe solved 57\% of the problems within 5 sec. For 92 benchmarks out
of 112, at least one PQE problem generated off \prob{\Abs{k-1}}{F_k}
was solved by \Apqe in the time limit. This is quite encouraging since
many solved PQE problems had more than a hundred thousand variables
and clauses.  Table~\ref{tbl:old_new} also shows that \Apqe
drastically outperforms \apqe.

To contrast PQE and QE, we used \cad~\cite{cadet_qe,cadet_imp} to
perform QE on 112 formulas \prob{\Abs{k-1}}{F_k}. That is, instead of
taking a clause out of \prob{\Abs{k-1}}{F_k} by PQE, we applied \cad
to perform full QE on this formula. (As mentioned in
Subsection~\ref{ssec:pg_by_pqe}, performing QE on
\prob{\Abs{k-1}}{F_k} produces the strongest local invariant
specifying all states unreachable in $k$ transitions.) Our choice of
\cad was motivated by its high performance. \cad is a SAT-based tool
that solves QE implicitly via building Skolem functions. In the
context of QE, \cad often scales better than
BDDs~\cite{bryant_bdds1,bryant_bdds2}.
\cad solved only 32 out of 112 QE problems with the time limit of 600
sec. For many formulas \prob{\Abs{k-1}}{F_k} that \cad failed to solve
in 600 sec., \Apqe solved all 50 PQE problems generated off
\prob{\Abs{k-1}}{F_k} in 5 sec. So, PQE can be much easier than QE if
only a small part of the formula gets unquantified.

%
%  Experiment 2
\subsection{Experiment 2}
\label{ssec:exper2}
%\tb{Experiment 2.}
The second experiment was an extension of the first experiment.
Namely, for each clause $Q$ of a local invariant $H_k$ generated by
PQE we used \ict to verify if $Q$ was a global invariant.  If so, we
checked if $\spe \not\imp Q$ held.

\begin{wraptable}{l}{3.1in}
%\begin{table}[h]
\centering
%  \small
\vspace{5pt}
\scriptsize
\caption{\small{A sample of HWMCC-13 benchmarks}}
%\caption{\small{A sample of HWMCC-13 benchmarks. The time limit is 5\,sec. per PQE problem}}
\vspace{2pt}
%\begin{center}
%\renewcommand{\arraystretch}{1.2} % Default value: 1
%\begni{tabular}{|p{22pt}|p{36pt}|c|c|c|c|} \hline
\begin{tabular}{|p{20pt}|p{20pt}|p{20pt}|p{20pt}|p{20pt}|p{20pt}|p{17pt}|p{12pt}|p{15pt}|p{20pt}|} \hline
  name     & lat- & invar.   &  time   &clau-   &\multicolumn{5}{c|}{local single-clause invariants} \\ \cline{6-10}
           & ches & of      &  fra-   &ses     & gen.    & \multicolumn{3}{c|}{global?} & not   \\ \cline{7-9}
           &      & \spe   &  mes    &taken   & inva-   & un-  & no  &  yes & impl.   \\
           &      &         &         &out   &  riants       & dec. &     &      & by \ti{Sp}  \\ \hline

6s380  & 5,606 & ~897 & ~~2 & ~46  & ~101  & ~0  & ~49 & ~52  & ~0 \\ \hline    
6s176  & 1,566 & ~952 & ~~3 & ~20  & ~101  & ~0  & ~9  & ~92  & ~14  \\ \hline
6s428  &3,790  & ~340 & ~~4 & ~29  & ~102  & ~15 & ~12 & ~75  & ~75 \\ \hline
6s292  &3,190  & ~247 & ~~5 & ~21  & ~104  & ~44 & ~0  & ~60  & ~60 \\ \hline
6s156  & 513   & ~32  & ~~6 & ~218 & ~101  & ~0  & ~90  & ~11  & ~11  \\ \hline
6s275  & 3,196 & ~673 & ~~7 & ~25  & ~106  & ~2  & ~21 & ~83  & ~77 \\ \hline
6s325  &1,756  & ~301 & ~~8 & ~23  & ~105  & ~0  & ~0  & ~105 & ~105 \\ \hline
6s391  &2,686  & ~387 & ~~9 & ~30  & ~104  & ~0  & ~14  & ~90 & ~90 \\ \hline
6s372  &1,124  & ~33  & ~10 & ~159 & ~101  & ~60 & ~41 & ~0  & ~0 \\ \hline
% &   &  &  &   & &    &    &    & &  \\ \hline
\end{tabular}                
%\end{center}
\vspace{5pt}
\label{tbl:sample}
%\end{table}
\end{wraptable}

Similarly to the first experiment, to make the experiment less time
consuming, we set the time limit of 5 sec. per PQE problem.  Besides,
we imposed the following constraints.  (Even with those constraints,
the run time of the experiment was about 4 days.) First, we stopped
\Apqe even before the time limit if it generated more than 5 free
clauses. Second, the time limit for \ict was set to 30 sec.  Third,
instead of constraining the number of PQE problems per benchmark, we
limited the total number of free clauses generated for a benchmark.
Namely, processing a benchmark terminated when this number exceeded
100.

A sample of 9 benchmarks out of the 112 we used in the experiment is
shown in Table~\ref{tbl:sample}. Let us explain the structure of this
table by the benchmark 6s380 (the first line of the table). The name
of this benchmark is shown in the first column. The second column
gives the number of latches (5,606). The number of invariants that
should hold for 6s380 is provided in the third column (897). So, the
specification \spe of 6s380 is the conjunction of those 897
invariants. The fourth column shows that the number $k$ of time frames
for 6s380 was set to 2 (since $|F_3|\!>\!500,000$).  The value 46
shown in the fifth column is the total number of clauses taken out of
\prob{\Abs{k-1}}{F_k} i.e. the number of PQE problems.  (We keep using
the index $k$ here assuming that $k\!=\!2$ for 6s380.)

Let $C$ be a clause taken out of the scope of quantifiers by
\Apqe. Every free clause $Q$ generated by \Apqe was stored as a local
single-clause invariant.  The sixth column shows that taking clauses
out of the scope of quantifiers was terminated when 101 local
single-clause invariants were generated.  (Because the total number of
invariants exceeded 100.) Each of these 101 local invariants held in
$k$-th time frame. The following three columns show how many of those
101 local invariants were true globally. \ict finished every problem
out of 101 in the time limit. So, the number of undecided invariants
was 0. The number of invariants \ict proved false or true globally was
49 and 52 respectively. The last column gives the number of global
invariants \ti{not} implied by \spe. For 6s380, this number is 0.

%\vspace{-5pt}
For 109 benchmarks out of the 112 we used in the experiments, \Apqe
was able to generate local single-clause invariants that held in
$k$-th time frame. For 100 benchmarks out of the 109 above, the
invariants $H_k$ generated by \Apqe contained global single-clause
invariants.  For 89 out of these 100 benchmarks, there were global
invariants not implied by the specification \spe. Those invariants
were meant to be checked if any of them was unwanted.
\clearpage
%\vspace{-10pt}
\subsection{Experiment 3}
\label{ssec:exper3}

When proving an invariant $P$, \ict conjoins it with clauses
$Q_1\!,\dots,\!Q_m$ to make $P\!\wedge Q_1\!\wedge \dots \wedge Q_m$
inductive.  If \ict succeeds, every $Q_i$ is an invariant. Moreover,
$Q_i$ may be an \ti{unwanted} invariant.  Arguably, the cause of
efficiency of \ict is that $P$ is often close to an inductive
invariant. So, \ict needs to generate a relatively small number of
clauses $Q_i$ to make the constrained version of $P$
inductive. However, as we mentioned in Section~\ref{sec:bg}, this nice
feature of \ict drastically limits the set of unwanted invariants it
can produce. In this subsection, we substantiate this claim by an
experiment. In this experiment, we picked the HWMCC-13 benchmarks for
which one could prove \ti{all} pre-defined invariants $P_1, \dots,
P_n$ within a time limit. Namely, for every benchmark we formed the
specification \mbox{\spe = $P_1 \wedge \dots \wedge P_n$} and ran \ict
to prove \spe true.

We selected the benchmarks that \ict solved in less than 1000 sec.
(In addition to dropping the benchmarks not solved in 1000 sec., we
discarded those where \spe failed because some invariants $P_i$ were
false). Let \bm{\spe^*} denote the inductive version of \spe produced
by \ict when proving \spe true.  That is, $\spe^*$ is \spe conjoined
with the invariant clauses generated by \ict.
For each of the selected benchmarks we generated invariants by \Apqe
exactly as in Experiment 2. That is, we stopped generation of local
single clause invariants when their number exceeded 100. Then we ran
\ict to identify local invariants that were global as well.  After
that we checked which of the global invariants generated by \Apqe were
not implied by \spe and $\spe^*$.

\begin{wraptable}{l}{2.5in}
%\begin{table}[h]
\centering
%  \small
\vspace{5pt}
\scriptsize
\caption{\small{Invariants of \Apqe and \ict}}
%\caption{\small{A sample of HWMCC-13 benchmarks. The time limit is 5\,sec. per PQE problem}}
\vspace{2pt}
%\begin{center}
%\renewcommand{\arraystretch}{1.2} % Default value: 1
%\begni{tabular}{|p{22pt}|p{36pt}|c|c|c|c|} \hline
\begin{tabular}{|p{23pt}|p{20pt}|p{20pt}|p{20pt}|p{20pt}|p{24pt}|} \hline
  name     & lat- & inva-  &\multicolumn{3}{c|}{glob sngl cls invars} \\ \cline{4-6}
           & ches & ri-    & glob.    & not & not   \\
           &      & ants    & inva-   & impl.   & impl.  \\
           &      & in \spe &  riants & by $\mi{Sp}$  &by  $\mi{Sp}^*$  \\ \hline

6s135  & 2,307  & 340 &  68 &   68   & 61 \\ \hline
6s325  & 1,756  & 301 & 101 &   101  & 96 \\ \hline
ex1    & 130    & 33  &  29 &   21   & 19 \\ \hline
ex2    & 212    & 32  &  93 &   61   & 42 \\ \hline
6s106  & 135    & 17  & 100 &   86   & 83 \\ \hline
6s256  & 3,141  & 5   &  0  &    0   &  0 \\ \hline
ex3    & 61     & 3   &  2  &   2    & 2  \\ \hline
ex4    & 63     & 3   &  3  &   3    & 3  \\ \hline
6s209  & 5,759  & 2   & 73  &   72   & 66 \\ \hline
6s113  & 994    & 1   &  18 &   17   & 17 \\ \hline
6s143  & 260    & 1   & 103 &   83   & 77 \\ \hline
6s170  & 3,141  & 1   &  1  &    1   &  1 \\ \hline
6s252  & 170   & 1   & 94  &   71   & 65 \\ \hline \hline
\tb{Total}  &       &     &     &  \tb{586}  & \tb{532} \\ \hline

\end{tabular}                
%\end{center}
\vspace{5pt}
\label{tbl:ic3_start}
%\end{table}
\end{wraptable}

The results of the experiment are shown in Table~\ref{tbl:ic3_start}.
The first three columns of this table are the same as in
Table~\ref{tbl:sample}. They give the name of a benchmark, the number
of latches and the number of invariants $P_1$,$\dots$,$P_n$ to
prove. (The actual names of examples \ti{ex1},..,\ti{ex4} in the
HWMCC-13 set are \ti{pdtvsarmultip}, \ti{bobtuintmulti},
\ti{nusmvdme1d3multi}, \ti{nusmvdme2d3multi}.) The next two columns of
Table~\ref{tbl:ic3_start} are the same as the last two columns of
Table~\ref{tbl:sample}. They show the number of local invariant
clauses that turn out to be global invariants and the number of global
invariants that were not implied by \spe.
The last column gives the number of global invariants that were not
implied by $\spe^*$. The last row of the table shows that in 532 cases
out of 586 the invariants not implied by \spe were not implied by
$\spe^*$ either. So, in 90\% of cases \Apqe generated invariant
clauses \ti{different} from those of \ict.

\end{document}